\documentclass[12pt,a4paper]{article}
\usepackage[truedimen,margin=30mm]{geometry} 

\usepackage{mathrsfs}
\usepackage{amssymb}
\usepackage{amsmath}
\usepackage{ascmac}
\usepackage{amsthm}
\usepackage[dvips]{graphicx}
\usepackage{natbib}
\usepackage{setspace}
\usepackage{times}
\usepackage{bm}

\usepackage{color}
\usepackage{subcaption}

\usepackage{comment}

\usepackage{titlesec}
\titleformat*{\section}{\large\bfseries}
\titleformat*{\subsection}{\it}

\newtheorem{thm}{Theorem}
\newtheorem{lem}{Lemma}

\newtheorem{prp}{Proposition}

\newcommand{\R}{\mathbb{R}}
\newcommand{\E}{\mathbb{E}}

\newcommand{\KL}{\mathrm{KL}}

\usepackage{enumitem}

\def\bx{{\boldsymbol{x}}}
\def\bm{{\boldsymbol{m}}}

\def\bI{{\boldsymbol{I}}}
\def\bX{{\boldsymbol{X}}}
\def\bY{{\boldsymbol{Y}}}
\def\bV{{\boldsymbol{V}}}
\def\bbeta{{\boldsymbol{\beta}}}
\def\veta{{\boldsymbol{\eta}}}
\def\veta{{\boldsymbol{\eta}}}

\newcommand{\normalphi}[3]{\phi(#1 | #2, #3)} 

\title{{\bf Repulsive $g$-Priors for Regression Mixtures}}

\date{}

\begin{document}

\maketitle
\doublespacing

\vspace{-1.5cm}
\begin{center}
{\large 
Yuta Hayashida$^1$ and Shonosuke Sugasawa$^2$\footnote{Corresponding author (Email: sugasawa@econ.keio.ac.jp)}
}

\medskip

\medskip
\noindent
$^1$Graduate School of Economics, Keio University\\
$^2$Faculty of Economics, Keio University\\
\end{center}

\vspace{0.5cm}
\begin{center}
{\bf \large Abstract}
\end{center}

\vspace{-0cm}
Mixture regression models are powerful tools for capturing heterogeneous covariate–response relationships, yet classical finite mixtures and Bayesian nonparametric alternatives often suffer from instability or overestimation of clusters when component separability is weak. Recent repulsive priors improve parsimony in density mixtures by discouraging nearby components, but their direct extension to regression is nontrivial since separation must respect the predictive geometry induced by covariates. We propose a repulsive $g$-prior for regression mixtures that enforces separation in the Mahalanobis metric, penalizing components indistinguishable in the predictive mean space. This construction preserves conjugacy-like updates while introducing geometry-aware interactions, enabling efficient blocked–collapsed Gibbs sampling. Theoretically, we establish tractable normalizing bounds, posterior contraction rates, and shrinkage of tail mass on the number of components. Simulations under correlated and overlapping designs demonstrate improved clustering and prediction relative to independent, Euclidean-repulsive, and sparsity-inducing baselines.

\bigskip\noindent
{\bf Key words}: mixture-of-experts; nonparametric Bayes; posterior consistency

\section{Introduction}
Mixture regression models are widely used to capture heterogeneous relationships across fields such as marketing, biostatistics, and econometrics \citep[e.g.][]{wedel1993latent,yao2010functional,hamilton2016macroeconomic}. Classical finite-mixture formulations, including hierarchical and adaptive mixtures of experts, have been extensively studied and are routinely employed \citep[e.g.][]{mclachlan2004finite,jordan1994hierarchical,jacobs1991adaptive}. However, these parametric models require fixing the number of components in advance and can be unstable when components overlap or covariates are high-dimensional. Bayesian nonparametric approaches such as Dirichlet process mixtures \citep{thomas1973bayesian,neal2000marcov} and priors on the number of components in finite mixtures, known as the mixture-of-finite-mixtures (MFM), offer more flexibility \citep{Miller02012018}. Yet, these approaches also face well-documented challenges for inference on cluster structure: Dirichlet process/Pitman–Yor mixtures can over-estimate the number of clusters, while overfitted finite mixtures tend to split components rather than emptying them unless the prior is carefully tuned \citep{miller2013inconsistency,xu2016bayesian}.

In mixture estimation, ensuring sufficient separability between components is often crucial for stable inference and interpretability. To this end, rather than using conventional independent priors for component-specific parameters, one can specify ``repulsive priors" as a joint prior that places low mass on nearby components, thereby discouraging redundant clusters and improving parsimony \citep{petralia2012repulsive,quinlan2017parsimonious}.
\cite{Xie02012020} developed a repulsive prior for Gaussian mixtures with theoretical guarantees, showing additional posterior shrinkage on the tail probability of the component count relative to independent priors. 
However, it cannot be directly imported to a regression setting, since the separation of the regression coefficients does not indicate the separation of the regression function when covariates are correlated or ill-conditioned.
Hence, the repulsive prior for the regression coefficients should take account of geometry induced by covariates.

To solve the aforementioned issue, we propose a repulsive $g$-prior for regression mixtures that enforces separation in the predictive geometry determined by the covariates. 
The prior measures pairwise distances between component coefficients in the Mahalanobis metric, whereby repulsion is strongest along well-identified directions and mild where the design is uninformative.
Equivalently, the repulsive $g$-prior discourages components that are nearly indistinguishable in the predictive mean space. 
Regarding its theoretical properties, adapting the normalization and tail–shrinkage arguments of \cite{Xie02012020} to our Mahalanobis penalty, we obtain a linear-in-$K$ bound on the normalizing constant and show shrinkage of posterior tail mass on the number of components. 
For posterior computation, we develop an efficient Gibbs sampler for the proposed repulsive $g$-prior, retaining conjugacy updates with a geometry-aware accept–reject step.

In related work, existing approaches primarily regulate model size or shrinkage without taking account of geometry induced by covariates. 
The complexity is usually controlled by independent component priors combined with either Dirichlet process mixtures or the mixture-of-finite-mixtures (MFM) prior \citep{Miller02012018, neal2000marcov}, or by sparsity-inducing priors on the weights that empty redundant components \citep{rousseau2011asymptotic}.
However, these mechanisms do not ensure separation in the predictive values. 
Repulsive priors for density mixtures instead penalize Euclidean proximity between component parameters \citep{petralia2012repulsive,Xie02012020}, but they are not tailored to regression geometry. 
More recent developments move beyond Euclidean isotropy, such as Wasserstein-based repulsion for density \citep{huang2025bayesian} and anisotropic repulsion in latent-factor clustering \citep{ghilotti2024bayesian}.
Yet these focus on density or latent-space structure rather than regression-specific predictive geometry.

The remainder of the paper is organized as follows. 
Section~2 introduces the regression‐mixture specification with the repulsive $g$-prior and develops a blocked-collapsed Gibbs sampler for posterior computation.
Section~3 establishes theoretical guarantees, where we prove strong posterior consistency, contraction rate, and quantify shrinkage of the posterior mass on the number of components.
Section 4 reports simulation studies to compare clustering and prediction performance of the proposed repulsive $g$-prior and other priors. 
Finally, Section 5 provides concluding remarks.

\section{Repulsive $g$-Priors}

\subsection{Model settings}

Let $y_i$ be a response variable and $\bx_i$ be a vector of covariates, for $i=1,\ldots,n$, where $n$ is the sample size.
We consider the following Gaussian regression mixture model:
\begin{equation}\label{eq:base_model}
f(y_i|\bx_i, \Theta)=\sum_{k=1}^K \omega_k \phi(y_i; \bx_i^\top \bbeta_k, \sigma_k^2),
\end{equation}
where $\omega_k$ is an unknown mixing proportion such that $\sum_{k=1}^K \omega_k=1$ with $\omega_k\geq 0$, and $\Theta=\{(\omega_k, \bbeta_k,\sigma_k),\ k=1,\ldots,K\}$ is a set of unknown parameters. 
Here $\phi(\cdot; \mu, \sigma^2)$ denotes the density function of the normal distribution with mean $\mu$ and variance $\sigma^2$. 
To allow model complexity to adapt to the data, we endow the number of components with a mixture-of-finite-mixtures (MFM) prior \citep{Miller02012018}, which places a discrete prior on $K$ and, conditional on $K$, assigns symmetric Dirichlet weights to the mixture proportions, thereby letting the posterior automatically infer an appropriate number of clusters.
Specifically, we assume that  
$$
(\omega_1, \ldots, \omega_K) \mid K \sim \mathcal{D}(\alpha,\ldots,\alpha), \qquad K \sim p(K), \qquad K \in \mathbb{N}_+,
$$
where $\mathcal{D}(\alpha,\ldots,\alpha)$ denotes the symmetric Dirichlet distribution. 

In the existing approaches, the regression coefficients, $\bbeta_1,\ldots,\bbeta_K$, are usually assumed independent across clusters. However, as shown in \cite{xu2016bayesian}, this standard assumption in mixture models often results in overlapping or redundant clusters, making the interpretation of the clusters challenging. 
To overcome this issue, we introduce repulsion to encourage distinct clusters in the mixture model, which gives a general form of the joint prior as follows: 
\begin{align*}
  p(\bbeta_1, \sigma_1, \ldots, \bbeta_K, \sigma_K \mid K) 
  = \frac{1}{Z_K}\left[\prod_{k=1}^{K}p_\beta(\bbeta_k)p_\sigma(\sigma_k)\right] h_K(\bbeta_1,\ldots,\bbeta_K)
\end{align*}
where $Z_K=\int\cdots\int h_K(\bbeta_1,\ldots,\bbeta_K)\big[\prod^K_{k=1}p_\beta(\bbeta_k)\big]d\bbeta_1\cdots d\bbeta_K$ is the normalization constant, and $h_K$ is a function that implements repulsion between the $\bbeta_k$. Notice that the repulsive prior defined here for regression mixtures is a simple extension of repulsive prior introduced in \cite{Xie02012020} for Gaussian mixture model. 
\cite{Xie02012020} have suggested $h_K(\bbeta_1,\ldots,\bbeta_K)= \min_{k<k'} G\bigl(\|\bbeta_k-\bbeta_{k'}\|\bigr)$,
where $G:\mathbb{R}_+\rightarrow[0,1]$ is strictly increasing with $G(0)=0$, such as $G(t)=t/(t+g_0)$, where $g_0>0$ is hyperparameter for repulsion. 
This form could apply repulsion via the Euclidean distances between pairs of coefficients, ensuring cluster separation.

Most existing repulsive priors have been developed for location–scale mixtures whose component means are given independent spherical normal priors. 
Extending them to regression mixtures is non-trivial because the design matrix $\bX=(\bx_1,\ldots,\bx_n)^\top$ induces non-spherical, data-dependent covariance structures.
Figure \ref{fig:x_xbeta_space} illustrates this difficulty. In  $\bbeta$-space, the Euclidean distances among three vectors, $\bbeta$, $\bbeta'$, $\bbeta''$, determine the strength of a spherical repulsive prior, with $\bbeta-\bbeta'$ attracting the strongest penalty and $\bbeta-\bbeta''$ the weakest. After projection to mean space via the linear map $f(\bbeta)=\bX\bbeta$, this ordering is reversed: $\bX\bbeta$ and $\bX\bbeta'$ become well separated, whereas $\bX\bbeta$ and $\bX\bbeta''$ nearly coincide. The example demonstrates that distance-based repulsion imposed in $\bbeta$-space fails to guarantee separation of component means once the design matrix is applied, motivating priors that measure repulsion directly in the induced $\bX\bbeta$-space.

\begin{figure}
\centering
\includegraphics[width=0.7\linewidth]{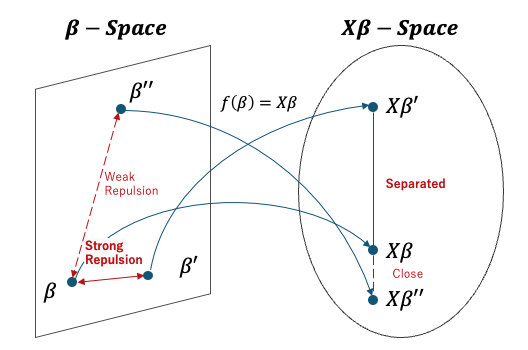}
\caption{ Distortion of repulsion under the design matrix.}
\label{fig:x_xbeta_space}
\end{figure}

To overcome the limitation, we derive the repulsive $g$-prior by applying the standard repulsive function to parameters transformed (whitened) via Zellner’s $g$-prior and subsequently reversing this transformation.
Consider constructing a prior for two regression coefficients, $\bbeta_1$ and $\bbeta_2$.
Throughout this paper, we assume that $\bX^\top \bX$ is non-singular. 
Then, the $g$-prior for $\bbeta_1$ and $\bbeta_2$ is defined as $N(0, g\sigma^2(\bX^\top \bX)^{-1})$.
This is equivalent to assuming that $\sqrt{g}(\bX^\top \bX)^{-1/2}\bbeta_k\sim N(0, \sigma^2 \bI_p)$ \ $(k=1,2)$, which means that the transformed parameter $\veta_k\equiv \sqrt{g}(\bX^\top \bX)^{-1/2}\bbeta_k$ can be treated as a multivariate parameter whose element following a independent prior with the same variance. 
Hence, instead of $\bbeta_1$ and $\bbeta_2$, we may consider a repulsive prior for $\veta_1$ and $\veta_2$ as follows: 
$$
\pi(\veta_1, \veta_2)\propto \phi(\veta_1; 0, \sigma^2 \bI_p)\phi(\veta_2; 0, \sigma^2 \bI_p)h(\|\veta_1-\veta_2\|^2).
$$
Note that $\|\veta_1-\veta_2\|^2=g(\bbeta_1-\bbeta_2)^\top (\bX^\top \bX)^{-1}(\bbeta_1-\bbeta_2)$, which can be regarded as the Mahalanobis distance between $\bbeta_1$ and $\bbeta_2$ with covariance matrix of $\bX$. 
Finally, the repulsive prior for $\bbeta_1$ and $\bbeta_2$ is obtained as 
$$
\pi(\bbeta_1, \bbeta_2)\propto \phi(\bbeta_1; 0, \sigma^2 A)\phi(\bbeta_2; 0, \sigma^2A)h((\bbeta_1-\bbeta_2)^\top A(\bbeta_1-\bbeta_2)),
$$
where $A=g(\bX^\top \bX)^{-1}$.
Then, the resulting joint prior is obtained as 
\begin{equation*}
\begin{split}
p(\bbeta_1, &\sigma_1, \ldots, \bbeta_K, \sigma_K \mid K) \\
  &= \frac{1}{Z_K}\left[\prod_{k=1}^{K}\phi(\bbeta_k;0, g\sigma^2(\bX^\top \bX)^{-1})p_\sigma(\sigma_k)\right] h_K(\bbeta_1,\ldots,\bbeta_K)
\end{split}
\end{equation*}
Because the $g$-prior shrinks each $\bbeta_k$ toward the origin in the metric induced by $\bX^{\top}\bX$, it provides scale-invariant and conjugate regularisation that meshes naturally with Gaussian likelihoods and already embeds information about the geometry of the covariates.  Building on this geometric insight, we depart from the usual practice of simply “plugging in” existing repulsive functions. Instead, we propose a novel form for the repulsive function that explicitly exploits the same $\bX^{\top}\bX$ metric:
\begin{align}
  h_K(\bbeta_1,\ldots,\bbeta_K)
  = \min_{k<k'} G\left((\bbeta_k-\bbeta_{k'})^{\top}g(\bX^\top \bX)^{-1}(\bbeta_k-\bbeta_{k'})\right)
\label{eq:proposed_repulsive}
\end{align}
so that the strength of repulsion between clusters directly reflects the covariance-shaped geometry of the predictors. This modification introduces explicit dependence on the design matrix $\bX$, in contrast to previous approaches. 
By imposing standard conditions on the distribution $p_x$, we can derive theoretical results for the relationship between $Z_K$ and $K$, similar to those established in \cite{Xie02012020}.

\begin{thm}
\label{thm:NormalizationConstant}
  Suppose each $\bbeta_k$ follows the $g$-prior and the form of repulsive function $h_K$ is given by (\ref{eq:proposed_repulsive}). If 
  \begin{align*}
    \iint\Bigl[\log G\!\bigl((\bbeta_1-\bbeta_2)^{\!\top}g(\bX^\top \bX)^{-1}(\bbeta_1-\bbeta_2)\bigr)\Bigr]^{2}p_\beta(\bbeta_1)p_\beta(\bbeta_2)\,d\bbeta_1d\bbeta_2<\infty,
  \end{align*}
  then $0 \leq -\log{Z_K} \leq c_1 K$ for some constant $c_1 > 0$.
\end{thm}

Theorem~\ref{thm:NormalizationConstant} shows that, under the metric induced by $\bX^\top\bX$, the normalizing constant $Z_K$ does not collapse and $-\log Z_K$ grows at most linearly in $K$. This ensures that the prior remains proper as the number of components increases and that the repulsive $g$-prior defines a valid probability distribution even for large $K$.

\subsection{Posterior computation}
We employ a blocked-collapsed Gibbs sampler for the Gaussian regression mixture (\ref{eq:base_model}) with the repulsive $g$-prior. 
To this end, we introduce a parameter of grouping assignment, denoted by $z_i\in \{1,2,\ldots\}$, for each subject. 
The detailed sampling steps are given as follows:

\begin{itemize}
\item 
(Precomputation of normalization constant)\ \ 
Fix $K_{\max}\!\ge\!2$. For $k=1,\ldots,K_{\max}$, define the prior-side normalizing constant
\[
Z_k
\;=\;
\int \cdots \int 
h_k(\bbeta_1,\ldots,\bbeta_k)\,
\Bigg[\prod_{j=1}^k p_\beta(\bbeta_j)\Bigg]\,
d\bbeta_1\cdots d\bbeta_k,
\]
where $p_\beta(\cdot)$ is the prior density for the regression coefficients. 
In practice, $Z_K$  is computed numerically using Monte Carlo integration. This is achieved by drawing a large number of independent samples, $\bbeta_j \sim p_\beta$, and then averaging the corresponding values of $h_k(\bbeta_1,\ldots,\bbeta_k)$.

\item
(Sampling of grouping assignment)\ \  First, generate auxiliary parameters as $\sigma_{\rm new}^{2} \sim \text{Inverse-Gamma}(a_0, b_0)$ truncated on $[\underline{\sigma}^2, \overline{\sigma}^2]$ and $\bbeta_{\rm new} \sim N(\mathbf{0}, g\sigma_{\rm new}^2(\bX^\top \bX)^{-1})$ and accept $\bbeta_{\rm new}$ with probability $h_{K+1}(\bbeta_1, \dots, \bbeta_K, \bbeta_{\rm new})$.
Then, for $i=1,\cdots,n$, the new assignment $z_i$ is generated from a multinomial distribution with probability 
\begin{align*}
&P(z_i=c|-) \propto (|c|+\alpha) \cdot \phi(y_i; \bx_i^\top\bbeta_c, \sigma_c^2), \quad  \ \ \  c \in \mathcal{C}_{-i}.\\
&P(z_i={\rm new}|-) \propto \frac{V_n(|\mathcal{C}_{-i}|+1)\alpha}{V_n(|\mathcal{C}_{-i}|)}
\cdot \phi(y_i; \bx_i^\top\bbeta_{\rm new}, \sigma_{\rm new}^2),
\end{align*}
where $V_n(t)$ denotes, 
\[
V_n(t)\ =\ \sum_{K=t}^{\infty} 
p_K(K)\ \frac{\Gamma(K+1)}{\Gamma(K-t+1)}\,
\frac{\Gamma(\alpha K)}{\Gamma(\alpha K + n)}.
\]
and $\mathcal{C}_{-i}$ is a set of group indices without the $i$th observation.

\item
(Sampling the number of components) \ \ 
Set $\ell=|\mathcal{C}|$ and consider candidate values $K\in\{\ell,\ell{+}1,\ldots,\ell{+}m\}$.
For each such $K$, define
\[
\widetilde Z_K
\ :=\
\int \cdots \int 
h_K\!\big(\{\bbeta_{c}\}_{c\in\mathcal{C}}\cup\{\bbeta_{c}\}_{c\in\mathcal{C}_\varnothing}\big)
\left[
  \prod_{c\in\mathcal{C}} p\!\left(\theta_c \mid \{y_i:\ i\in c\},\bX\right)
\right]
\left[
  \prod_{c\in\mathcal{C}_\varnothing} p\!\left(\theta_c\right)
\right]
\,d\theta,
\]
where $\theta_c=(\bbeta_c,\sigma_c^2)$ and $\mathcal{C}_\varnothing$ indexes the $K-\ell$ empty clusters, and compute numerically by Monte Carlo.
Then sample $K$ from the discrete posterior using the precomputed $Z_K$ as
$$
p\left(K \mid -\right)
\ \propto\
\frac{\widetilde Z_K}{Z_K}\,
\frac{K!}{(K-\ell)!\,(K+n)!},
\qquad 
K\in\{\ell,\ell{+}1,\ldots,\ell{+}m\}.
$$

\item
(Sample cluster-wise variance) \ \ 
For $c=1,\ldots,K$, generate $\sigma_c^2$ from its full conditional posterior
$$
\sigma_c^2\sim \text{Inverse-Gamma}\left(a_0 + \frac{|c|}{2}, b_0 + \frac{1}{2} \sum_{z_i\in \mathcal{C}} (y_i - \bx_i^\top\bbeta_c)^2 \right).
$$

\item
(Sampling of cluster-wise coefficients) \ \ 
For $c=1,\ldots,K$, generate a proposal $\bbeta_{c}'$ from its full conditional distribution as follows: 
\begin{itemize}
\item[-]
For non-empty cluster, generate $\bbeta_c'$ from its full conditional posterior $N(\bV_c \bm_c, \bV_c)$, where
\begin{align*}
\bV_c = \left(\frac{1}{\sigma_c^{2}}\bX_c^\top \bX_c + \frac{1}{g\sigma_c^{2}}\bX^\top \bX\right)^{-1},  \ \ \ \ \ 
\bm_c = \frac{1}{\sigma_c^{2}}\bX_c^\top \bY_c,
\end{align*}
where $\bX_c$ and $\bY_c$ are sub-matrix and sub-vector of $\bX$ and $\bY$ satisfying $z_i=c$, respectively. 

\item[-]
For empty cluster, generate $\beta_{c}'$ from its $g$-prior $N(\mathbf{0}, g\sigma_{c}^2(\bX^\top \bX)^{-1})$.
\end{itemize}
The proposal $\bbeta_c'$ \ $(c=1,\ldots,K)$ is accepted with probability $h_K(\bbeta_1', \dots, \bbeta_K')$.

\end{itemize}

\section{Theoretical Properties}

Based on the theoretical framework by \cite{Xie02012020}, we extend the theoretical analysis of Bayesian repulsive mixture models to the context of regression mixtures with a repulsive $g$-prior. 
Furthermore, we analyze the shrinkage effect of the repulsive prior on the posterior of the number of components $K$, highlighting both the technical and practical advantages of the proposed prior in regression mixture.

\subsection{Assumptions}
The first set of conditions are requirements for the true distribution ($f_0$, $F_0$, $p_X$) and the general structure.

\begin{enumerate}[label=\textbf{A\arabic*.}, noitemsep, topsep=0cm]
    \item The true mixing distribution $F_0$ on $\Theta = \R^p \times [\underline{\sigma}^2, \overline{\sigma}^2]$ has a sub-Gaussian tail for the regression coefficients $\bbeta$: $\int \|\bbeta\|^k dF_0(\bbeta, \sigma^2) < \infty$ for all $k \ge 1$.
    \item The function $G$ used in the repulsive function $h_K$ satisfies: for some $\delta_g > 0, c_g > 0$, we have $G(x) \ge c_g \epsilon$ whenever $x \ge \epsilon$ and $\epsilon \in (0, \delta_g)$.
    \item The function $G$ and the base prior $p_\beta$ satisfy the integrability condition required for the bound on the normalizing constant $Z_K$:
    $$ 
    \iint_{\R^p \times \R^p} [\log G((\bbeta_1-\bbeta_2)^\top g(\bX^\top \bX)^{-1}(\bbeta_1-\bbeta_2))]^2 p_\beta(\bbeta_1) p_\beta(\bbeta_2) d\bbeta_1 d\bbeta_2 < \infty.
    $$
    \item The true mixing distribution $F_0$ has support for $\sigma^2$ contained within known bounds: there exist $0 < \underline{\sigma}^2 \le \overline{\sigma}^2 < \infty$ such that $\text{supp}(F_0(\cdot, \sigma^2)) \subset [\underline{\sigma}^2, \overline{\sigma}^2]$. We also assume the prior $p_{\sigma^2}$ has the same support $[\underline{\sigma}^2, \overline{\sigma}^2]$ .
    \item The true covariate density $p_X(\bx)$ is bounded, and $x$ satisfies  $\|\bx\|_2 \le M_X< \infty$.
\end{enumerate}

Assumption A1 restricts the true mixing distribution to have sub-Gaussian tails, ruling out excessively heavy-tailed coefficients. Assumption A2 imposes minimal regularity on the repulsive function $G$ to ensure sufficient separation between components. Assumption A3 guarantees that the normalization constant $1/Z_K$ does not grow super-exponentially with $K$. Assumption A4 requires both the true and prior variances to be bounded away from zero and infinity. Finally, Assumption A5 bounds the covariate distribution, which simplifies entropy calculations in the later theoretical analysis.

We also need some requirements for the prior distributions $\Pi$ over the mixing measure $F = \sum_{k=1}^K w_k \delta_{(\bbeta_k, \sigma_k^2)}$.

\begin{enumerate}[resume, label=\textbf{A\arabic*.}, noitemsep, topsep=0cm]
    \item The prior on weights is $(w_1, \dots, w_K | K) \sim \mathcal{D}_K(\alpha)$ with $\alpha \in (0, 1]$.
    \item The base prior density $p_\beta(\bbeta)$ for the regression coefficients has a sub-Gaussian tail: $\int_{\{\|\bbeta\| \ge t\}} p_\beta(\bbeta) d\bbeta \le B_2 e^{-b_2 t^2}$ for some $B_2, b_2 > 0$ .
    \item The base prior density $p_\beta(\bbeta)$ is positive and continuous everywhere on $\R^p$: $p_\beta(\bbeta) > 0$ for all $\bbeta \in \R^p$.
    \item The prior density $p_{\sigma^2}(\sigma^2)$ has support $[\underline{\sigma}^2, \overline{\sigma}^2]$ and is positive and continuous on its support.
    \item The prior on the number of components $p_K(K)$ decays sufficiently fast but not too fast for large $K$. There exist $B_4, b_4 > 0$ such that for sufficiently large $K$:
    $$ p_K(K) \ge e^{-b_4 K \log K}, \quad \sum_{N=K}^{\infty} p_K(N) \le e^{-B_4 K \log K} $$
\end{enumerate}

Assumption A6 assumes a weakly informative Dirichlet prior for the component weights, which is standard and ensures adequate flexibility. Assumption A7 requires the base prior for the regression coefficients to 
have sub-Gaussian tails, preventing the prior from concentrating on extreme values. 
Assumption A8 ensures that the prior for $\bbeta$ is everywhere positive and continuous, so that all regions of the parameter space are accessible. 
Assumption A9 requires the prior for the noise variance to be positive and continuous within its support. 
Finally, Assumption A10 controls the tail behavior of the prior on the number of components, ensuring it neither decays too slowly nor too quickly as $K$ increases.

\subsection{Consistency and contraction rate}


To establish the strong consistency of the proposed model, we follow the general approach in Theorem 1 of \cite{canale2017posterior}, suitably adapted to the regression mixture setting. Specifically, we construct a sequence of sieve submodels of $\mathcal{M}(\mathbb{R}^p \times \mathbb{R}_+)$ defined by
\[
\mathcal{F}_{K_n}=\left\{f_F(\cdot|\cdot):F= \sum_{k=1}^{K} \omega_k \delta_{(\beta_k, \sigma_k)},K\le K_n,\beta_k\times\sigma_k^2\in\R^p\times\R_+\right\}
\]
and the following partition of the submodel $\mathcal{F}_{K_n}$
\[
\mathcal{G}_K(a_K)=\mathcal{F}_{K}\left(\prod_{k=1}^K(a_k,a_k+1]\right),\,a_K=(a_1,\cdots,a_K)\in\mathbb{N}^K,\,K=1,\cdots K_n,
\]
where
\[
\mathcal{F}_{K}\left(\prod_{k=1}^K(a_k,b_k]\right)=\left\{f_F(\cdot|\cdot):F= \sum_{k=1}^{K} \omega_k \delta_{(\beta_k, \sigma_k)},\|\beta\|_\infty\in(a_k,b_k]\right\}
\]
This construction generalizes the partitions introduced in \cite{Xie02012020} to the regression mixture case, allowing for control over the complexity of the parameter space.
According to Theorem 1 of \cite{canale2017posterior}, it suffices to verify two main conditions: (i) the true density $f_0$ is in the Kullback–Leibler (KL) support of the prior $\Pi$; and (ii) there exist $\tilde{b}>0$ and a sequence $(K_n)_{n=1}^\infty$ such that for sufficiently large $n$, the following summability condition holds for all $\epsilon > 0$:
\begin{align}
\label{eq:summability_condition}
    \lim_{n\to\infty}e^{-(4-\tilde{b})n\epsilon^2}
\sum_{K=1}^{K_n}\sum_{a_1=0}^{\infty}\dots\sum_{a_K=0}^{\infty}
\sqrt{\mathcal{N}(\epsilon,\mathcal G_K(\mathbf a_K),\|\cdot\|_{1})}
   \sqrt{\Pi(\mathcal G_K(\mathbf a_K))}=0.
\end{align}
where $\mathcal{N}(\epsilon, \mathcal{G}_K(\mathbf{a}_K), \|\cdot\|_1)$ denotes the $\epsilon$-covering number of $\mathcal{G}_K(\mathbf{a}_K)$ under the $L^1$ norm.
As shown in the Supplementary Materials, we can verify all required conditions of Theorem 1 in \cite{canale2017posterior} for our regression mixture model with a repulsive $g$-prior. 
Then, we obtain the following strong consistency result:

\begin{thm}
Under Assumptions A1-A10, the posterior $\Pi(\cdot|y_1,\cdots,y_n)$ is strongly consistent at $f_0$.
\end{thm}

We next consider the posterior contraction rate for the regression mixture model with the repulsive $g$-prior. 
The result is as follows: 

\begin{thm}
Under Assumptions A1-A10, the posterior distribution $\Pi(\cdot\mid y_1,\cdots,y_n)$ contracts at $f_0$ with rate $\epsilon_n=(\log n)^{t}/\sqrt{n}$, $t>p+\frac{\alpha+2}{4}$.
\end{thm}

Thus, the above property guarantees that the repulsive prior does not adversely affect the rate of posterior contraction relative to standard location mixture models, while still encouraging separation between mixture components in the regression setting.

\subsection{Shrinkage effect on the posterior of $K$}

While the previous sections have established the theoretical soundness of the proposed regression mixture model with a repulsive $g$-prior, an important practical advantage of the repulsive prior lies in its ability to control model complexity by preventing the overestimation of the number of components $K$. In many applications, standard mixture models with independent priors tend to allocate redundant or overlapping clusters, which can lead to unnecessarily large values of $K$ and complicate interpretation.

To formalize this advantage, we analyze the shrinkage effect of the repulsive prior on the posterior distribution of $K$. Our analysis extends Theorem 4 of \cite{Xie02012020}, which originally demonstrated this shrinkage phenomenon for location mixture models, to the context of regression mixtures with a $g$-prior structure. The result below quantifies how the repulsive prior penalizes excessive clustering through the normalization constant, leading to a tighter posterior distribution on $K$.

\begin{thm}
\label{thm:ShrinkageRate}
Assume a mixture of regressions model where the error variance is fixed to $\sigma_0^2$. Let the base prior for the regression coefficients $\boldsymbol{\beta}_k$ be the $g$-prior, $p(\boldsymbol{\beta}) = N(\boldsymbol{\beta} | \boldsymbol{0}, g\sigma_0^2(\bX^\top \bX)^{-1})$, and the repulsive function be $h_K(\boldsymbol{\beta}_1, \dots, \boldsymbol{\beta}_K) = \min_{k<k'} G(d_M(\boldsymbol{\beta}_k, \boldsymbol{\beta}_{k'}))$, where $d_M$ is the squared Mahalanobis-like distance and $G(d) = d/(g_0+d)$ for some $g_0 \ge 0$. 
Assume the prior on the number of components is $p(K) \propto Z_K \lambda^K/K!$. 
Let the true data generating process be $f_0(\boldsymbol{y}|\boldsymbol{X}) = \int \prod_{i=1}^n \phi(y_i|x_i^\top\boldsymbol{m}_i, \sigma_0^2) dF_0(\boldsymbol{m}_1, \dots, \boldsymbol{m}_n)$. Then, for sufficiently large $N$, the expected posterior tail probability of $K$ satisfies the following inequality: 
$$
\mathbb{E}_{f_0}[\Pi(K \ge N | \boldsymbol{y}, \boldsymbol{X})] \le C(\lambda, \boldsymbol{X}) \cdot \chi(g_0, \boldsymbol{X}, n, N) \sum_{K=N+1}^{\infty} \frac{\lambda^K}{(e^\lambda - 1)K!}
$$
where $C(\lambda, \boldsymbol{X})$ is a constant. The shrinkage constant $\chi(g_0, \boldsymbol{X}, n, N)$ is given by
$$
\chi(g_0, \boldsymbol{X}, n, N) = (1+\delta(g,X)g_{0}^{2/3})^{3/2} \cdot G\left(\sqrt{\frac{2n}{N}\mathbb{E}_{F_0}[\boldsymbol{m}^\top (g^{-1}(\boldsymbol{X}^\top\boldsymbol{X})) \boldsymbol{m}] + C_1}\right)
$$
where $\delta(g,X)$ is a constant depending on $(g,X)$ such that $\delta(g,X)<1$ for sufficiently large $g$, and satisfies $\chi(0, \cdot) = 1$ and $\chi(g_0, \cdot) < 1$ for $g_0>0$.
\end{thm}
Theorem \ref{thm:ShrinkageRate} shows that the presence of the repulsive prior ($g_0 > 0$) leads to a shrinkage factor $\chi(g_0, \cdot) < 1$, resulting in a posterior that is more concentrated around smaller values of $K$. In contrast, when $g_0 = 0$ (i.e., no repulsion), $\chi(0, \cdot) = 1$ and no shrinkage occurs, often leading to persistent overestimation of the number of clusters in practice. This result highlights a key benefit of the proposed approach: by explicitly penalizing overly similar regression coefficients through the geometry of the $g$-prior, the model encourages parsimony and interpretability in the inferred clustering structure.

\section{Simulation Study}

\subsection{Illustration of repulsive $g$-prior }

Before turning to Monte Carlo simulation studies, we first present a qualitative illustration of the proposed prior. 
To this end, we set $K_{\mathrm{true}} = 3$ (true number of clusters) and $n=3000$ (sample size) with equal sample size across clusters. 
The $i$th observation assigned in cluster $k$ is generated as  
$$
y_i=\beta_{k,0} + \beta_{k,1}x_i+\epsilon_i,
$$
where $x_i\sim \text{Uniform}(0,10)$, $\epsilon_i\sim\mathcal N(0,1)$, and coefficient across clusters, $\beta_k=(\beta_{k,0}, \beta_{k,1})$, are set to $\beta_1=(-5.0,2.5), \beta_2=(0.0,1.0)$ and $\beta_3=(-1.0,1.5)$.
For the generated dataset, we fit three models, standard mixture-of-finite-mixtures (MFM), MFM with standard repulsive prior (RRM), and 
the proposed repulsive $g$-prior (RgRM), described as follows:
\begin{itemize}
\item[-] (RgRM: Repulsive $g$-prior regression mixture) \ 
Consider a finite mixture with an over–specified number of components $K_{\text{fit}}$. Let $\pi$ lie on the simplex and $(\beta_k,\sigma_k^2)$ index component $k$. We place Zellner’s $g$-prior on the coefficients,
$$
\beta_k \mid \sigma_k^2 \sim \mathcal N\!\left(0,\; g\,\sigma_k^2 \,(X^\top X)^{-1}\right),
$$
with $g=n$. A repulsive potential acts on pairs of components through a Mahalanobis metric in coefficient space. The hyperparameter $g_0=1$ sets the strength of repulsion. This model targets well-separated regression clusters while preserving scale adaptivity through the $g$-prior. We use number of occupied components as a number of clusters.

\item[-] (RRM: Repulsive regression mixture) \ 
We keep the same finite mixture structure and repulsive mechanism. We replace the $g$-prior with a normal prior, $\beta_k \sim \mathcal N(0,\tau^2 I_p)$.
Repulsion acts through a Euclidean metric on $\beta$. The hyperparameter $g_0=1$ again controls the strength.

\item[-] (MFM: Mixture of finite mixtures with non-repulsive prior) \ 
This model removes repulsion. It keeps the normal prior for $\beta_k$ with variance $\tau^2$. All other ingredients match the repulsive standard model. This baseline shows how much the repulsive term contributes on its own.
\end{itemize}
We fix $K_{\text{fit}}=20$ (the maximum number of clusters) for all models.
For standard normal priors, we set $\tau^2=1$. 
For each mode, we generate $1000$ posterior samples after discard the first $1000$ as burn-in.

Figure~\ref{fig:toy_results} shows cluster assignments and the size of clusters obtained by each method. 
It is observed that the standard MFM over-estimates the number of components ($\hat{K}=9$) by dividing the true three clusters into multiple clusters and it produces clusters with small observations. 
On the other hand, the standard repulsive prior under-estimates the number of components ($\hat{K}=2$) by merging the two different clusters into a single cluster. 
On the other hand, the proposed method recovers the correct number of clusters ($\hat{K}=3$), and provides reasonable clustering structures of three regression functions. 
This qualitative illustration highlights the motivation and advantages of the proposed method.

\begin{figure}[htbp]
\centering 
\includegraphics[width=\linewidth]{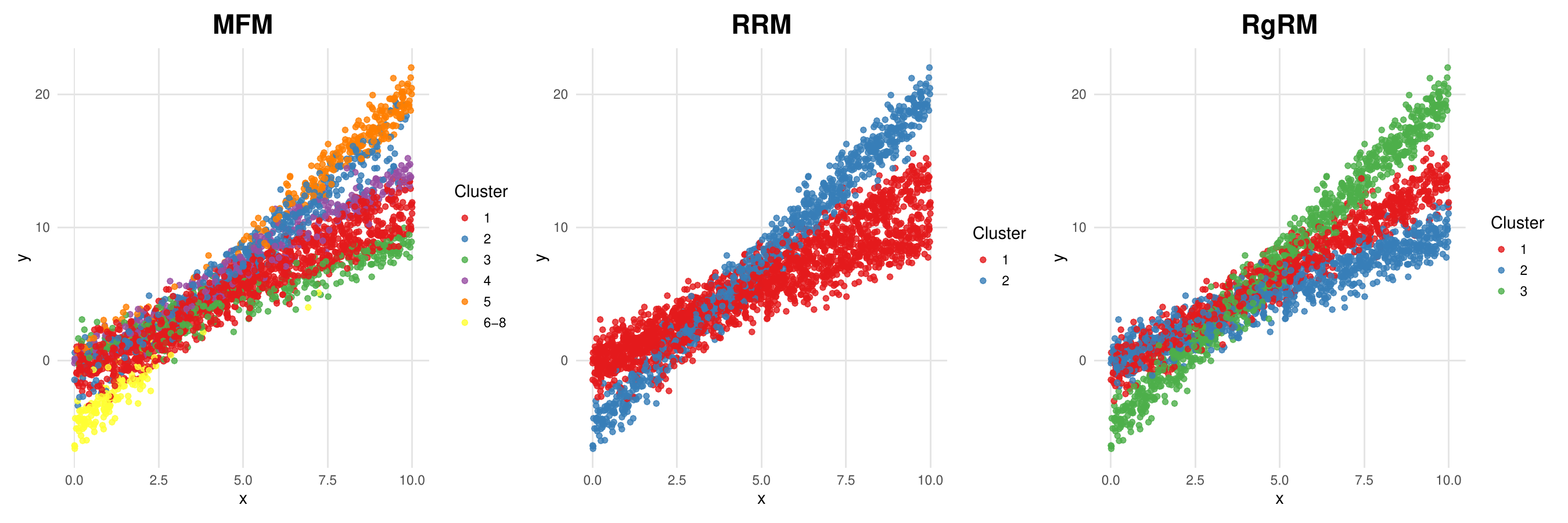}
\includegraphics[width=\linewidth]{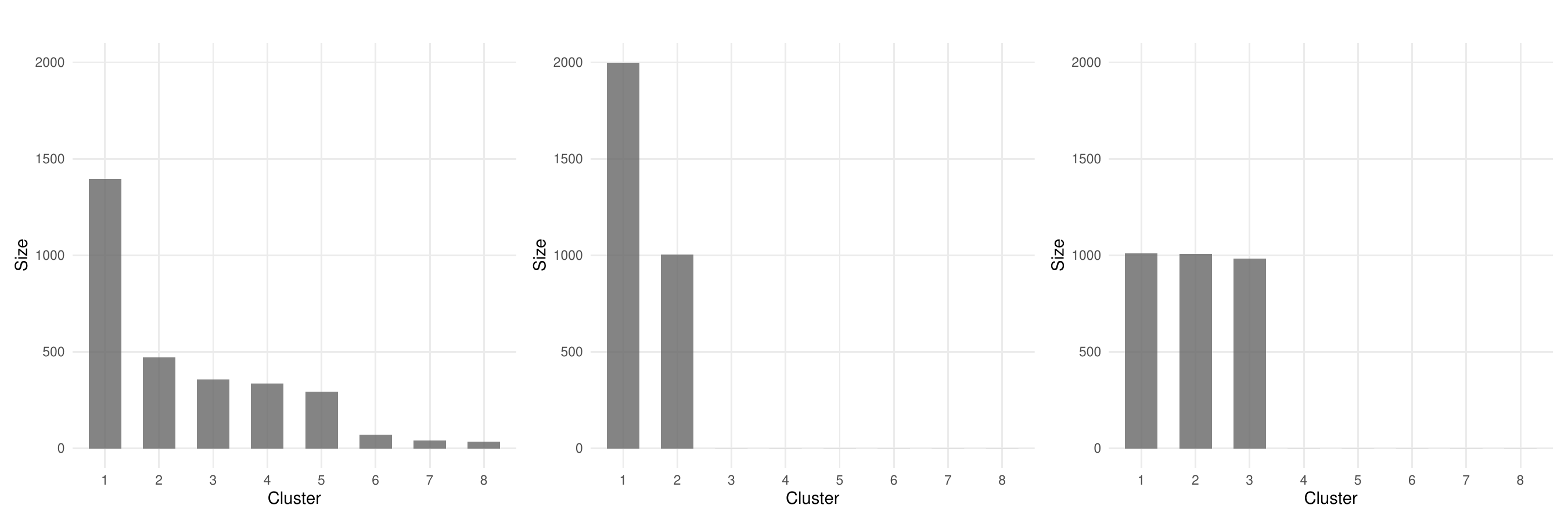}
\caption{Scatter plots with clustering results (upper) and the number of clusters and cluster sizes (lower) obtained by the three methods, under one-shot simulation data.   }
\label{fig:toy_results}
\end{figure}

\subsection{Monte Carlo simulations: data generation and methods}
To evaluate the performance of the proposed model under different data characteristics, we consider three distinct simulation scenarios. These scenarios investigate how robustly each model performs across varying conditions of feature scaling imbalance, feature correlation, and cluster distinctiveness.

We study a mixture of linear regressions with four clusters and four covariates. 
Each dataset contains $n=4\times n_{\text{per}}$ observations with equal allocation across clusters. For observation $i$ in cluster $k$,
$$
y_i=\beta_{k,1}x_{i,1}+\cdots+\beta_{k,4}x_{i,4}+\epsilon_i,\qquad \epsilon_i\sim\mathcal N(0,1).
$$
We draw covariates independently from a multivariate Normal law with mean vector of ones and covariance $\Sigma_X$. All clusters share the same $\Sigma_X$. Coefficients are fixed at the cluster centers $\{\beta_k\}_{k=1}^4$. We consider three scenarios that change $\Sigma_X$ and the coefficient matrix. We set $n_{\text{per}}\in\{25,50,100\}$, which gives total sample sizes $n\in\{100,200,400\}$.
We considered the following three scenarios: 
\begin{itemize}
\item[-] {\bf (Scenario 1: balanced scale and no correlation)} \ 
This baseline uses independent covariates with common scale. Let $\Sigma_X=\mathrm{diag}(10^2,\;10^2,\;10^2,\;10^2)$. Set $\beta_{k,j}=I\{j=k\}$ for $k=1,\ldots,4$.
Each cluster activates a distinct covariate, which yields clear separation.

\item[-] {\bf (Scenario 2: scale imbalance with correlation)} \ 
Here the fourth covariate is much larger in scale and the features are moderately correlated. Let $\Sigma_{i,j}= \sigma_i\sigma_j\,\rho^{|i-j|}$ with $(\sigma_1,\sigma_2,\sigma_3,\sigma_4)=(10,10,10,100)$ and $\rho=0.5$. Keep the diagonal pattern for $\beta$ but downweight the high–variance direction in the fourth cluster: $\beta_{k,j}=1$ for $j=k,\ k\neq 4$, $\beta_{k,j}=0.01$ for $j=k=4$, and $\beta_{k,j}=0$ otherwise.   
This choice reduces the signal to noise along the largest scale.

\item[-] {\bf (Scenario 3: overlapping coefficients)} \
We use the same $\Sigma_X$ as in Scenario 2. 
Coefficients now overlap across clusters as $\beta_1=(1,1,0,0)$, $\beta_2=(0,1,1,0)$, $\beta_3=(0,0,1,0.01)$, $\beta_4=(1,0,0,0.01)$.

Two clusters share adjacent effects, and the fourth covariate contributes only weakly. This setting creates close clusters and frequent ambiguities.
\end{itemize}
Across scenarios we fix the error variance at one and omit within–cluster coefficient noise. Thus each cluster is defined by a fixed $\beta_k$ and the common covariate law. The design isolates the roles of scale, correlation, and coefficient overlap on mixture recovery.

We compare the three Bayesian mixture regressions, as described in the previous section, and also consider the following finite mixture regression: 
\begin{itemize}
\item[-] (SID: Finite mixture with sparsity-inducing Dirichlet priors) \ 
SID keeps $K_{\text{fit}}$ components but pushes many weights toward zero through a small-mass Dirichlet prior \citep{rousseau2011asymptotic}.
Let $\pi \sim \mathrm{Dirichlet}(\alpha K_{\rm fit}^{-1},\ldots,\alpha K_{\rm fit}^{-1})$, $\beta_k \sim \mathcal{N}(0,\tau^2 I_p)$ and $\sigma_k^2 \sim \mathrm{InvGamma}(a_0,b_0)$.
Empty or nearly empty components receive near-zero weight, so the posterior favors a small effective number of clusters, $K_{\text{eff}}=\#\{k:\pi_k>10^{-3}\}$.  We study two levels of concentration, ``high" and ``low" by setting $\alpha$ at two fixed values, $\alpha=0.1$ (SID1) and $\alpha=0.02$ (SID2).
\end{itemize}

For SID we set $(a_0,b_0)=(4,4)$ and study two values of $\alpha$. 
For the posterior computation of SID, we used the blocked sampler in \cite{rousseau2011asymptotic}.
As in the previous section, we generated 1000 posterior samples after discarding the first 1000 samples as burn-in, for each simulated data.

Based on posterior samples, we compute cluster labels by $\arg\max_k P(z_i=k\mid y_i,x_i)$. 
For $\hat{K}$, we use $K_{\text{eff}}$ for SID and the number of occupied components for the other models. 
For prediction, we post-hoc refit an ordinary least squares estimator within each estimated cluster using the observations assigned to that cluster and set $\hat y_i=\bx_i^\top \tilde{\bbeta}_{\hat z_i}$, which removes prior-specific shrinkage effects and makes predictions comparable across priors.
Before computing metrics, we address label switching by enforcing a consistent labeling convention across posterior samples.
We then summarize performance with adjusted rand index (ARI), Purity, estimated number of cluster $\hat{K}$ and root mean squared errors (RMSE). 
ARI and Purity are defined in  \cite{manning2008introduction}. 
Finally, RMSE is defined as $\mathrm{RMSE}= \sqrt{n^{-1}\sum_{i=1}^n \bigl(y_i-\hat y_i\bigr)^2 }$, where $\hat y_i=\bx_i^\top \hat{\bbeta}_i$ with $\hat{\bbeta}_i=\bbeta_{z_i}$  in MFM based models and  $\hat{\bbeta}_i=\bar{\bbeta}_{z_i}$ in SID models where $\bar{\beta_k}$ represents the posterior mean coefficient vector for component $k$.
For each scenario, we generated $200$ independent datasets and compute each metric on every dataset and report the averaged values across replications, with the Monte Carlo standard errors.

\subsection{Monte Carlo simulation: results}

Table \ref{tab:results} reports the results under $n=100$ and $n=200$, where the results with $n=400$ are provided in the Supplementary Material. 
Across all conditions, the proposed RgRM offers the most stable balance between accuracy and prediction error and keeps $\hat K$ close to the true value $(K=4)$. 
RRM is competitive when separation is generous but its error grows as designs become more demanding. 
MFM overestimates the component count by a wide margin and performs poorly on accuracy. SID1 and SID2 tend to underestimate components in the harder settings, which is reflected in weaker clustering and purity.

In the independent and balanced design, RgRM and RRM match the accuracy of SID1 and SID2, yet only RgRM keeps $K$ calibrated and holds error down. MFM inflates $K$, which indicates that control of $K$ requires a repulsive mechanism.
With correlated covariates and with coefficient overlap, the differences sharpen. RgRM keeps $K$ stable and maintains low error. RRM underestimates $K$ because strong repulsion merges clusters that are far in prediction but close in coefficient space. Both RgRM and RRM outperform SID1 and SID2, which lack repulsion between clusters.
Overall, enforcing separation in the $X^\top X$ geometry yields calibrated component counts and a favorable accuracy-error trade off in balanced designs and under scale imbalance, correlation, and coefficient overlap. The gains are most pronounced when separation is limited, where RgRM keeps $\hat K$ near four and avoids the large prediction errors observed in competing approaches.

\begin{table}[htb!]
\centering
 \caption{Average values of adjusted rand tndex (ARI), purity, estimated number of cluster ($\hat{K}$) and root mean squared errors (RMSE), based on 200 Monte Carlo replications under $n=100$ and $n=200$. 
 The Monte Carlo standard errors are given in the parenthesis. }
\label{tab:results}
\begin{tabular}{ccccccccccccccccccc}
\hline
Scenario  & $n$ & Method &  & ARI & RMSE & $\hat{K}$ & Purity \\
\hline
 &  & RgRM &  & 0.55 {\scriptsize (0.13)} & 1.14 {\scriptsize (0.15)} & 4.10 {\scriptsize (0.41)} & 0.80 {\scriptsize (0.08)} \\
 &  & RRM  &  & 0.57 {\scriptsize (0.12)} & 1.19 {\scriptsize (1.00)} & 3.98 {\scriptsize (0.44)} & 0.81 {\scriptsize (0.09)} \\
1 & 100 & MFM   &  & 0.01 {\scriptsize (0.02)} & 2.76 {\scriptsize (1.20)} & 10.3 {\scriptsize (2.13)} & 0.43 {\scriptsize (0.04)} \\
 &  & SID1 &  & 0.62 {\scriptsize (0.09)} & 0.91 {\scriptsize (0.45)} & 4.33 {\scriptsize (0.52)} & 0.84 {\scriptsize (0.05)} \\
 &  & SID2 &  & 0.58 {\scriptsize (0.12)} & 1.33 {\scriptsize (1.21)} & 3.85 {\scriptsize (0.43)} & 0.81 {\scriptsize (0.09)} \\
 \hline
 &  & RgRM &  & 0.42 {\scriptsize (0.13)} & 1.19 {\scriptsize (0.23)} & 4.10 {\scriptsize (0.33)} & 0.72 {\scriptsize (0.09)} \\
 &  & RRM  &  & 0.38 {\scriptsize (0.16)} & 2.00 {\scriptsize (1.42)} & 3.54 {\scriptsize (0.69)} & 0.68 {\scriptsize (0.14)} \\
2 & 100 & MFM   &  & 0.05 {\scriptsize (0.04)} & 5.09 {\scriptsize (1.16)} & 6.55 {\scriptsize (1.91)} & 0.42 {\scriptsize (0.06)} \\
 &  & SID1 &  & 0.18 {\scriptsize (0.18)} & 4.04 {\scriptsize (1.89)} & 2.39 {\scriptsize (1.08)} & 0.45 {\scriptsize (0.19)} \\
 &  & SID2 &  & 0.06 {\scriptsize (0.10)} & 5.41 {\scriptsize (1.29)} & 1.48 {\scriptsize (0.69)} & 0.33 {\scriptsize (0.11)} \\
 \hline
 &  & RgRM &  & 0.30 {\scriptsize (0.15)} & 1.66 {\scriptsize (0.37)} & 4.23 {\scriptsize (0.63)} & 0.62 {\scriptsize (0.11)} \\
 &  & RRM  &  & 0.32 {\scriptsize (0.18)} & 3.31 {\scriptsize (2.01)} & 3.03 {\scriptsize (0.83)} & 0.60 {\scriptsize (0.16)} \\
3 & 100 & MFM   &  & 0.01 {\scriptsize (0.02)} & 7.31 {\scriptsize (2.09)} & 8.81 {\scriptsize (1.99)} & 0.40 {\scriptsize (0.04)} \\
 &  & SID1 &  & 0.26 {\scriptsize (0.20)} & 4.27 {\scriptsize (2.41)} & 2.69 {\scriptsize (1.13)} & 0.52 {\scriptsize (0.19)} \\
 &  & SID2 &  & 0.08 {\scriptsize (0.12)} & 6.55 {\scriptsize (1.76)} & 1.49 {\scriptsize (0.69)} & 0.34 {\scriptsize (0.13)} \\
 \hline
 &  & RgRM &  & 0.61 {\scriptsize (0.06)} & 1.04 {\scriptsize (0.06)} & 4.00 {\scriptsize (0.00)} & 0.83 {\scriptsize (0.03)} \\
 &  & RRM  &  & 0.61 {\scriptsize (0.06)} & 1.00 {\scriptsize (0.07)} & 4.03 {\scriptsize (0.16)} & 0.84 {\scriptsize (0.03)} \\
1 & 200 & MFM   &  & 0.01 {\scriptsize (0.01)} & 2.34 {\scriptsize (0.75)} & 14.1 {\scriptsize (1.83)} & 0.41 {\scriptsize (0.03)} \\
 &  & SID1 &  & 0.65 {\scriptsize (0.06)} & 0.88 {\scriptsize (0.05)} & 4.12 {\scriptsize (0.33)} & 0.86 {\scriptsize (0.03)} \\
 &  & SID2 &  & 0.65 {\scriptsize (0.07)} & 1.02 {\scriptsize (0.63)} & 3.97 {\scriptsize (0.24)} & 0.85 {\scriptsize (0.05)} \\
 \hline
 &  & RgRM &  & 0.49 {\scriptsize (0.07)} & 1.07 {\scriptsize (0.56)} & 3.99 {\scriptsize (0.22)} & 0.78 {\scriptsize (0.05)} \\
 &  & RRM  &  & 0.46 {\scriptsize (0.12)} & 1.52 {\scriptsize (1.25)} & 3.80 {\scriptsize (0.49)} & 0.74 {\scriptsize (0.11)} \\
2 & 200 & MFM   &  & 0.05 {\scriptsize (0.04)} & 5.06 {\scriptsize (1.16)} & 8.62 {\scriptsize (2.60)} & 0.41 {\scriptsize (0.06)} \\
 &  & SID1 &  & 0.35 {\scriptsize (0.18)} & 3.10 {\scriptsize (1.79)} & 2.93 {\scriptsize (1.01)} & 0.60 {\scriptsize (0.18)} \\
 &  & SID2 &  & 0.12 {\scriptsize (0.12)} & 5.21 {\scriptsize (1.09)} & 1.67 {\scriptsize (0.64)} & 0.37 {\scriptsize (0.12)} \\
 \hline
 &  & RgRM &  & 0.51 {\scriptsize (0.10)} & 1.14 {\scriptsize (0.23)} & 4.15 {\scriptsize (0.49)} & 0.78 {\scriptsize (0.07)} \\
 &  & RRM  &  & 0.40 {\scriptsize (0.17)} & 2.81 {\scriptsize (1.77)} & 3.33 {\scriptsize (0.69)} & 0.67 {\scriptsize (0.15)} \\
3 & 200 & MFM   &  & 0.02 {\scriptsize (0.01)} & 7.08 {\scriptsize (1.60)} & 11.2 {\scriptsize (1.60)} & 0.39 {\scriptsize (0.03)} \\
 &  & SID1 &  & 0.36 {\scriptsize (0.19)} & 3.89 {\scriptsize (2.11)} & 2.80 {\scriptsize (0.94)} & 0.59 {\scriptsize (0.17)} \\
 &  & SID2 &  & 0.19 {\scriptsize (0.17)} & 5.62 {\scriptsize (1.96)} & 1.99 {\scriptsize (0.85)} & 0.44 {\scriptsize (0.16)} \\
\hline
\end{tabular}
\end{table}

\section{Concluding Remarks}

This paper introduces a new repulsive prior called ``repulsive $g$-priors" for mixture of regression models, which enforces separation in the predictive geometry induced by covariates and enables efficient posterior computation with geometry-aware Gibbs sampling. 
Our theoretical analysis establishes tractable bounds, posterior contraction, and shrinkage of the posterior tail mass, and simulations demonstrate improved clustering and predictive performance relative to existing priors.
While the current study has focused on standard linear regression mixtures, a natural direction for future research is to extend the repulsive $g$-prior framework to generalized linear models. 
In such cases, the definition of repulsion will require metrics beyond the Mahalanobis distance, tailored to the information-geometric structure of the chosen link and variance functions. 
Developing principled priors under these alternative geometries, together with computational strategies for efficient MCMC implementation, represents an important avenue for future work.

\section*{Acknowledgement}
This work is partially supported by JSPS KAKENHI Grant Numbers 24K21420 and 25H00546.

\vspace{1cm}
\bibliographystyle{chicago}
\bibliography{ref}

\clearpage
\setcounter{equation}{0}
\setcounter{section}{0}
\setcounter{table}{0}
\setcounter{page}{1}
\setcounter{lem}{0}
\setcounter{prp}{0}
\renewcommand{\thesection}{S\arabic{section}}
\renewcommand{\theequation}{S\arabic{equation}}
\renewcommand{\thetable}{S\arabic{table}}
\renewcommand{\thelem}{S\arabic{lem}}
\renewcommand{\theprp}{S\arabic{prp}}

\vspace{1cm}
\begin{center}
{\LARGE
{\bf Supplementary Material for ``Repulsive $g$-Priors for Regression Mixtures"}
}
\end{center}

\bigskip
This Supplementary Material provides the proofs of the theorems presented in the main text, and additional simulation results.

\section*{Notations}
We begin by introducing the notation and mathematical objects used throughout the theoretical analysis. This section closely follows the convention in \cite{Xie02012020}, with necessary modifications for the regression mixture setting.

Let $\mathcal{M}(\Theta)$ denote the space of all probability measures on the parameter space $\Theta$. In our case, $\Theta = \mathbb{R}^p \times [\underline{\sigma}^2, \overline{\sigma}^2]$ represents the set of possible regression coefficients and noise variances for each component in the mixture model. This provides the foundational setting for both the true and the estimated mixing distributions.
For a metric space $(\mathcal F, d)$, the $\varepsilon$-covering number $N(\varepsilon, \mathcal F, d)$ is the minimal number of $d$-balls of radius $\varepsilon$ required to cover $\mathcal F$. 
The (metric) entropy $\log N(\varepsilon, \mathcal F, d)$ quantifies the complexity of function classes and plays a crucial role in bounding the covering numbers of model sieves when establishing posterior consistency and contraction rates.
For any $\bx \in \mathbb{R}^p$, let $f_F(y, \bx) = p_X(\bx) \int \phi(y \mid \bx^\top\bbeta, \sigma^2) dF(\bbeta, \sigma^2)$ denote the (joint) observational model induced by a mixing distribution $F$ and covariate distribution $p_X$. The true data-generating process is assumed to be $f_0(y, \bx)$ of the same form, but governed by the true mixing distribution $F_0$.
We denote $\phi(y \mid \bx^\top\bbeta, \sigma^2)$ as the Gaussian kernel with mean $\bx^\top\bbeta$ and variance $\sigma^2$. For any two densities $f, g$ over $(y, \bx)$, the Kullback–Leibler (KL) divergence is defined as $D_{\KL}(f | g) = \int f(\bx) \log \{f(\bx)/g(\bx)\}dx$. 
Throughout, $|\cdot|_1$ refers to the $L^1$ norm over the density functions.
The above notation will be used throughout to define model sieves, partitions, and KL-type neighborhoods, which are essential for verifying the conditions required for posterior consistency and contraction results.

\section{Proof of Theorem 1}

Let $A=g(X^{\top}X)^{-1}$ and define the transformed coefficients $\eta_k=A^{1/2}\beta_k$.
Because $A$ is symmetric positive-definite, we have
$$
(\beta_1-\beta_2)^{\!\top}A(\beta_1-\beta_2)=\|\,\eta_1-\eta_2\|^{2}.
$$
The Jacobian of the change of variables is $|\det A^{1/2}|$, so, setting
$$
p_\eta(\eta)=p_\beta(A^{1/2}\eta)\,|\det A^{1/2}|,
$$
we obtain
$$
\begin{aligned}
&\iint\Bigl[\log G\!\bigl((\beta_1-\beta_2)^{\!\top}A(\beta_1-\beta_2)\bigr)\Bigr]^{2}p_\beta(\beta_1)p_\beta(\beta_2)\,d\beta_1d\beta_2 \\
&\quad=\iint\Bigl[\log G\bigl(\|\eta_1-\eta_2\|^{2}\bigr)\Bigr]^{2}
p_\eta(\eta_1)p_\eta(\eta_2)\,d\eta_1d\eta_2<\infty.
\end{aligned}
$$
Thus the integrability condition required in  \cite{Xie02012020} (Assumption A2) holds for the independent and identically distributed random variables $\eta_1,\dots,\eta_K$.
Applying Theorem A.1 of \cite{Xie02012020} to $\{\eta_k\}$ yields $0\le -\log Z_{K}\le c_1K$ for some finite $c_1$.

%
\section{Proof of Theorem~2}

It is sufficient to verify the conditions of Theorem A.1 in \cite{canale2017posterior}, which mainly consists of three lemmas given below. 
First, we show that the true density $f_0$ is in the KL-support of the prior under mild regularity conditions:

\begin{lem}
\label{lemma:KLsupport}
Let $f_0(y, x) = p_X(x) \int_{\Theta} \phi(y; x^\top\beta, \sigma^2) dF_0(\beta, \sigma^2)$ be the true density, where $\Theta = \R^p \times [\underline{\sigma}^2, \overline{\sigma}^2]$. 
Define the truncated parameter space $$T_m = \{(\beta, \sigma^2) \in \Theta : \|\beta\| \le m, \underline{\sigma}^2 + 1/m \le \sigma^2 \le \overline{\sigma}^2 - 1/m \}$$ for $m \ge m_0$ such that $F_0(T_{m_0}) > 0$.
Let $F_m(A) = c_mF_0(A \cap T_m)$ and $c_m^{-1}=F_0(T_m)$.
Thus, $f_{F_m}(y, x) = p_X(x) c_m\int_{T_m} \normalphi{y}{x^\top\beta}{\sigma^2} dF_m(\beta, \sigma^2)$.
Then, under A1 and A4, it holds that 
$$ \lim_{m\to\infty} D_{\KL}(f_0 \| f_{F_m}) = \lim_{m\to\infty} \E_{P_0}\left[ \log \frac{f_0(Y|X)}{f_{F_m}(Y|X)} \right] = 0 $$
\end{lem}

\begin{proof}
Without loss of generality, we assume that $T_1$ is non-empty. Clearly, $T_m \uparrow \Theta$ and $c_m \downarrow 1
$ as $m\to\infty$ by the monotone continuity of the probability measure $F_0$. 
Furthermore,  $\normalphi{y}{x^\top\beta}{\sigma^2}\le (2\pi\underline{\sigma}^2)^{-1/2}$ .
Hence, for fixed $x$, $$f_{F_m}(y\mid x) = c_m\int_{T_m} \normalphi{y}{x^\top\beta}{\sigma^2} dF_m(\beta, \sigma^2) \to \int_{\Theta} \normalphi{y}{x^\top\beta}{\sigma^2} dF_0(\beta, \sigma^2)= f_0(y\mid x)$$ by the bounded convergence theorem, implying that $\log \frac{f_0(y\mid x)}{f_{F_m}(y\mid x)} \to \log 1 = 0$ .
In order to show $\lim_{m\to\infty} \iint f_0(y|x) \log \frac{f_0(y|x)}{f_{F_m}(y|x)} dy = 0$, it suffices to find a dominating function $g(y|x)$ such that $\left| \log \frac{f_0(y|x)}{f_{F_m}(y|x)} \right| \le g$ for all $m \in \mathbb{N}_{+}$ , and the conclusion is guaranteed by the dominating convergence theorem
First of all, notice that for all $m \in \mathbb{N}_{+}$ , we have $f_{F_m}(y|x) \le c_m \int \phi dF_0 \le  c_1 (2\pi\underline{\sigma}^2)^{-1/2}$, and thus $f_0\le c_1(2\pi\underline{\sigma}^2)^{-1/2}$ . It follows that $\log \frac{f_0(y|x)}{f_{F_m}(y|x)} \ge \log \frac{f_0(y|x)}{c_1 (2\pi\underline{\sigma}^2)^{-1/2}}$ .
Next, we see that,
\begin{align*}
f_{F_m}(y|x) &= c_m \int_{T_m} \normalphi{y}{x^\top\beta}{\sigma^2} dF_0(\beta, \sigma^2) \\
&\ge \int_{T_1} \normalphi{y}{x^\top\beta}{\sigma^2} dF_0(\beta, \sigma^2) \quad  \\
&\ge \int_{T_1} (2\pi\overline{\sigma}^2)^{-1/2} \exp\left(-\frac{\|y - x^\top \beta\|^2}{2\underline{\sigma}^2}\right) dF_0(\beta, \sigma^2) \quad 
\end{align*}
On $T_1$, we have $\|\beta\| \le 1$. Using the Cauchy-Schwarz inequality, $\|x^\top \beta\| \le \|x\| \|\beta\| \le \|x\|$. Then bound is $\|y - x^\top \beta\|^2 \le (2 \max\{\|y\|, \|x^\top \beta\|\})^2 \le (2 \max\{\|y\|, \|x\|\} )^2$.
It follows that
\begin{align*} 
f_{F_m}(y|x)\ge \xi(y|x)  := (2\pi\overline{\sigma}^2)^{-1/2} \exp\left(-\frac{2 (\max\{\|y\|, \|x\|\})^2}{\underline{\sigma}^2}\right) F_0(\{(\beta,\Sigma)\in T_1\}) \\
\end{align*}
and thus, $\log \frac{f_0(y|x)}{f_{F_m}(y|x)} \le \log \frac{f_0(y|x)}{\xi(y|x)}$ . In particular, letting $m\to\infty$, $f_0(y|x) \ge \xi(y,x)$.
Together we have
$$ \log \frac{f_0(y|x)}{c_1 (2\pi\underline{\sigma}^2)^{-1/2}} \le \log \frac{f_0(y|x)}{f_{F_m}(y|x)} \le \log \frac{f_0(y|x)}{\xi(y|x)} $$
which implies
$$ \left| \log \frac{f_0(y|x)}{f_{F_m}(y|x)} \right| \le g(y|x):= \max\left\{ \left|\log \frac{f_0(y|x)}{c_1 (2\pi\underline{\sigma}^2)^{-1/2}}\right|, \left|\log \frac{f_0(y|x)}{\xi(y|x)}\right| \right\} $$
To show that $g$ is $f_0$-integrable, it suffices to verify the $f_0$-integrability of $\log f_0(y|x)$ and $\log \xi(y|x)$. 
Notice that $\log\xi\le\log{f_0}\le \log(c_1(2\pi\underline{\sigma}^2)^{-1/2})=\log(c_1)-\frac{1}{2}\log(2\pi\underline{\sigma}^2)$ , implying
$$
|\log f_0| \le |\log c_1|+ \frac{1}{2}|\log(2\pi\underline{\sigma}^2)|+|\log\xi|,
$$
it is only left to verify the $f_0$-integrability of $\log\xi$.
\begin{align*}
\int& f_0(y|x) |\log \xi(y|x)|  dy\\
    &=-\frac{1}{2}(2\pi\overline{\sigma}^2)+|\log F_0(\{(\beta,\Sigma)\in T_1\})|+\frac{2}{\underline{\sigma}^2}\int f_0(y|x) (\max\{\|y\|, \|x\|\})^2  dy\\
    &\le |\log F_0(\{(\beta,\Sigma)\in T_1\})|+\frac{2}{\underline{\sigma}^2}\int f_0(y|x) (\|y\|^2 + \|x\|^2)  dy\\
   & = |\log F_0(\{(\beta,\Sigma)\in T_1\})|+\frac{2}{\underline{\sigma}^2}(\E_0[y^2|x] +  \|x\|^2)\\
   &<\infty
\end{align*}
where the finiteness of $\E_0[y^2|x]$ is guaranteed by condition A1 and Fubini’s theorem and finiteness of $\|x\|^2$ is guaranteed by condition A5 for any $x$. Hence, $\log\xi$ is $f_0$-integrable, which implies that $g(y|x)$ is integrable.
By the dominated convergence theorem,
$$ \lim_{m\to\infty} D_{\KL}(f_0 \| f_{F_m}) = \int f_0(y|x) \lim_{m\to\infty} \left( \log \frac{f_0(y|x)}{f_{F_m}(y|x)} \right) dy = \int f_0(y|x) \times 0 dy = 0 $$
The proof is thus completed.
\end{proof}

Next, we provide a bound for the covering number of the sieves, ensuring the complexity of each $\mathcal{F}_K$ remains manageable:

\begin{lem}
\label{lemma:CoveringNumber}
Let $a_k<b_k$ be non-negative integers, $k=1,\cdots,K$. Then for sufficiently small $\delta>0$, there exists constant $c_4$ such that
    $$\mathcal{N}\left(\delta, \mathcal{F}_K\left(\prod_{k=1}^K(a_k,b_k]\right), \|\cdot\|_{1}\right) \nonumber
\leq \left(\frac{c_4 M_X^p}{\delta^{d+2}}\right)^K \left(\prod_{k=1}^K b_k\right)^p $$
\end{lem}

\begin{proof}
Suppose $\delta>0$ is given.
By condition A5, covariates $x$ satisfy $\|x\|_2\le M_X$ for some constant $M_X<\infty$.
By Lemma A.4 in \cite{ghosal2001entropies}, there exists an $l_1$ $\delta$-net $\mathcal{I}_0$ of $\Delta^K$, such that the cardinality $|\mathcal{I}_0|$ of $\mathcal{I}_0$ is upper bounded by $(5/\delta)^K$.\\
Now let $\mathcal{R}_k$ be an $\ell_\infty$-net for the regression coefficients $\mathbf{\beta}_k \in \mathbb{R}^p$ within the region $\{ \mathbf{\beta}_k : \|\mathbf{\beta}_k\|_\infty \in (a_k, b_k] \}$. Let $\Delta_\beta =c_\beta \delta/\sqrt{p}M_X$ for some constant $c_\beta>0$, then $|\mathcal{R}_k|\le(b_k/\Delta_\beta + 1)^p=(b_k\sqrt{p}M_X/c_\beta\delta + 1)^p$.\\
Furthermore let $\mathcal{S}_{k}$ be an $\delta$-net for the variance $\sigma_k^2$ in the interval $[\underline{\sigma}^2, \overline{\sigma}^2]$. The cardinality is $|\mathcal{S}_{k}|\leq(\overline{\sigma}^2- \underline{\sigma}^2)/\delta+1 $.\\
It follows that for all $f_F(y|x) \in \mathcal{F}_K\left(\prod_{k=1}^K(a_k,b_k]\right)$ with $ F=\sum_{k=1}^K w_k \delta_{(\mathbf{\beta}_k, \sigma_k^2)},$
there exists some $\mathbf{w}^\star=(w_1^\star,\dots,w_K^\star)\in\mathcal{I}_0$, $\mathbf{\beta}_k^\star\in\mathcal{R}_k$, $(\sigma_k^\star)^2 \in \mathcal{S}_{k}$ for $k=1,\dots,K$, such that $\sum_{k=1}^K|w_k-w_k^\star|<\delta$,
$\|\mathbf{\beta}_k - \mathbf{\beta}_k^\star\|_\infty \le \Delta_\beta$ (implying $\|\mathbf{\beta}_k - \mathbf{\beta}_k^\star\|_2 \le \sqrt{p}\Delta_\beta = c_\beta\delta/M_X$), and $|\sigma_k^2 - (\sigma_k^\star)^2| \le \delta$.\\
Denote $H(f,g)$ to be the Hellinger distance.
We evaluate the Hellinger distance between regression kernels $\phi(y|x^\top\beta_k, \sigma_k^2)$ and $\phi(y|x^\top\beta_k^\star, (\sigma_k^\star)^2)$. For a fixed $x$, let $\mu_k(x) = x^\top\beta_k$ and $\mu_k^\star(x) = x^\top\beta_k^\star$. The squared Hellinger distance is
\begin{align*}
& H^2(\phi(\cdot|\mu_k(x), \sigma_k^2), \phi(\cdot|\mu_k^\star(x), (\sigma_k^\star)^2)) \nonumber\\
&= 1 - \sqrt{\frac{2\sigma_k \sigma_k^\star}{\sigma_k^2 + (\sigma_k^\star)^2}} \exp\left(-\frac{(\mu_k(x) - \mu_k^\star(x))^2}{4(\sigma_k^2 + (\sigma_k^\star)^2)}\right) \nonumber\\
&\leq\left(1-\sqrt{\frac{2\sigma_k \sigma_k^\star}{\sigma_k^2 + (\sigma_k^\star)^2}}\right)+\left(1-\exp\left(-\frac{(\mu_k(x) - \mu_k^\star(x))^2}{4(\sigma_k^2 + (\sigma_k^\star)^2)}\right)\right)\\
&\leq\left(1-\sqrt{1-\frac{(\sigma_k-\sigma_k^*)^2}{\sigma_k^2+(\sigma_k^*)^2}})\right)+\left(1-\exp\left(-\frac{(c_\beta\delta)^2}{4(2\underline{\sigma}^2)}\right)\right)\\
&\leq\frac{(\sigma_k-\sigma_k^*)^2}{2\underline{\sigma}^2}+\frac{(c_\beta\delta)^2}{4(2\underline{\sigma}^2)}\quad (\because \forall x\in[0,1], 1-\sqrt{1-x}=\frac{x}{1+\sqrt{1-x}}\le x)\\
&=\left(\frac{1}{8\underline{\sigma}^2}+\frac{c_\beta^2}{8\underline{\sigma}^2}\right)\delta^2:=C_1\delta^2
\end{align*}
where $|\mu_k(x) - \mu_k^\star(x)| = |x^\top(\beta_k - \beta_k^\star)| \le \|x\|_2 \|\beta_k - \beta_k^\star\|_2 \le M_X (c_\beta\delta/M_X) = c_\beta\delta$.

Denote $ F^\star=\sum_{k=1}^K w_k^\star \delta_{(\mathbf{\beta}_k^\star, (\sigma_k^\star)^2)}$.
It follows by the triangle inequality that
\begin{align*}
\|f_F - f_{F^\star}\|_1 &= \iint \| p_X(x) \sum w_k \phi_k(x) - p_X(x) \sum w_k^\star \phi_k^\star(x) \| dy dx \\
&= \int p_X(x) \| \sum w_k \phi_k(x) - \sum w_k^\star \phi_k^\star(x) \|_{1,y} dx \\
&\le \int p_X(x) \left( \sum_{k=1}^K |w_k - w_k^\star| |\phi_k(x)|_{1,y} + \sum_{k=1}^K w_k^\star |\phi_k(x) - \phi_k^\star(x)|_{1,y} \right) dx \\
&\le \sum_{k=1}^K |w_k - w_k^\star| \int p_X(x) dx + \sum_{k=1}^K w_k^\star \int p_X(x) 2\sqrt{2} H(\phi_k(x), \phi_k^\star(x)) dx \\
&\le \delta + \sum_{k=1}^K w_k^\star \cdot 2\sqrt{2}\sqrt{C_1} \delta \int p_X(x) dx  \\
&= \delta + 2\sqrt{2}\sqrt{C_1} \delta \sum w_k^\star =  (1+2\sqrt{2}\sqrt{C_1})\delta := C_2\delta.
\end{align*}
Since $ |\phi_k(x)|_{1,y} = 1 \text{ and } |\phi_k(x) - \phi_k^\star(x)|_{1,y} \le 2\sqrt{2} H(\phi_k(x), \phi_k^\star(x))$.
therefore, 
\begin{align*}
&\mathcal{N}\left(C_2\delta, \mathcal{F}_K\left(\prod_{k=1}^K(a_k,b_k]\right), \|\cdot\|_{1}\right) \nonumber\\
&\leq |\mathcal{I}_0| \cdot \prod_{k=1}^K |\mathcal{R}_k| \cdot \prod_{k=1}^K |\mathcal{S}_k| \nonumber\\
&\leq \left(\frac{5}{\delta}\right)^{K} \times \prod_{k=1}^K \left(\frac{\sqrt{p} b_k M_X}{c_\beta\delta+1}\right)^p \times \prod_{k=1}^K \left(\frac{\overline{\sigma}^2- \underline{\sigma}^2}{\delta+1}\right) \nonumber\\
&\leq \left(\frac{5}{\delta}\right)^{K} \times \prod_{k=1}^K \left(\frac{\sqrt{p} b_k M_X}{c_\beta\delta}\right)^p \times \prod_{k=1}^K \left(\frac{\overline{\sigma}^2- \underline{\sigma}^2}{\delta}\right) \nonumber\\
&= \left(\frac{5(\sqrt{p}M_X/c_\beta)^p(\overline{\sigma}- \underline{\sigma})}{\delta^{1+p+1}}\right)^K \left(\prod_{k=1}^K b_k\right)^p \\
&= \left(\frac{c_3 M_X^p}{\delta^{p+2}}\right)^K \left(\prod_{k=1}^K b_k\right)^p
\end{align*}
for some constant $c_3>0$. This yields that
\begin{align}
&\mathcal{N}\left(\delta, \mathcal{F}_K\left(\prod_{k=1}^K(a_k,b_k]\right), \|\cdot\|_{1}\right) \nonumber
\leq \left(\frac{c_4 M_X^p}{\delta^{d+2}}\right)^K \left(\prod_{k=1}^K b_k\right)^p \nonumber
\end{align}
for some constant $c_4>0$.
\end{proof}

Finally, we derive a bound for the covering number of the sieves, ensuring the complexity of each $\mathcal{F}_K$ remains manageable:

\begin{lem}
\label{lemma:SumUpperBound}
Assume conditions A1-A10 hold. Then we have
$$
\sum_{K=1}^{K_n}\sum_{a_1=0}^\infty\dots\sum_{a_K=0}^\infty \sqrt{\mathcal{N}\left(\delta,\mathcal{G}K(\mathbf{a}K), \|\cdot\|_{1}\right)} \sqrt{\Pi\left(\mathcal{G}K(\mathbf{a}K)\right)}\le K_n \left( \frac{M}{\delta^{(d+2)/2}} \right)^{K_n}. 
$$
for sufficiently small $\delta$ for some constant $M>0$.
\end{lem}

\begin{proof}
First we need to bound $\Pi(\mathcal{G}_K(a_K))$, where $\mathcal{G}_K(a_K)=\mathcal{F}_K(\Pi_{k=1}^KI_j(a_k))$ and $I_j(a_k)={\beta\in \R^p:\|\beta\|_\infty\in(a_k,a_k+1]}$.
Recall that $Z_K \geq e^{-c_1K}$ for some constant $c_1>0$ by Theorem 1.
It holds that 
\begin{align*}
\Pi\left(\mathcal{G}K(\mathbf{a}K)\right)
&\le \Pi\left(\forall k: \|\mathbf{\beta}_k\|_\infty > a_k \mid K\right) p_K(K)\\
&\le \frac{p_K(K) }{Z_K}\int\cdots\int \prod_{k=1}^K (\mathbb{I}(\|\mathbf{\beta}_k\|_\infty \ge a_k)p_\beta(\mathbf{\beta}k))  d\mathbf{\beta}_ 1\cdots d\mathbf{\beta}_ K \nonumber \\
&\le e^{c_1K} \prod_{k=1}^K \int_{\|\beta_k\|\ge a_k} p_\beta(\mathbf{\beta}_k) d\mathbf{\beta}_k \quad \text{(since } |\beta_k|\infty > a_k \implies |\beta_k| \ge a_k \text{ )}\nonumber\\
&\le e^{c_1K} \prod_{k=1}^K \left( B_2 e^{-b_2 a_k^2} \right) \quad \text{(by Assumption A7)} \nonumber\\
&= e^{c_1K} B_2^K \prod_{k=1}^K \exp\left(-b_2 a_k^2\right).
\end{align*}

Now by Lemma~\ref{lemma:CoveringNumber} for some constant $c_4>0$, we have
\begin{align*}
\mathcal{N}(\delta, \mathcal{G}K(\mathbf{a}K), \|\cdot\|_{1}) \le \left(\frac{c_4 M_X^p}{\delta^{d+2}}\right)^K \prod_{k=1}^K (a_k+1)^p. 
\end{align*}
Hence, by defining $S=\sum_{a_k=0}^\infty(a_k+1)^{p/2}\exp(-b_2a_k^2/2)<\infty $(since $b_2>0$), we estimate
\begin{align*}
\sum_{K=1}^{K_n}&\sum_{a_1=0}^\infty\dots\sum_{a_K=0}^\infty \sqrt{\mathcal{N}\left(\delta,\mathcal{G}K(\mathbf{a}K), \|\cdot\|_{1}\right)} \sqrt{\Pi\left(\mathcal{G}K(\mathbf{a}K)\right)} \nonumber\\
&\le\sum_{K=1}^{K_n}\sum_{a_1=0}^\infty\dots\sum_{a_K=0}^\infty \left[ \frac{\sqrt{c_4} M_X^{p/2}}{\delta^{(p+2)/2}} \right]^K\left[ \prod_{k=1}^{K}(a_k+1)^{p/2} \right]\left[ \sqrt{e^{c_1}B_2}\right]^K\left[\prod_{k=1}^{K}\exp\left(-\frac{b_2 a_k^2}{2}\right)\right]\\
&= \sum_{K=1}^{K_n}\sum_{a_1=0}^\infty\dots\sum_{a_K=0}^\infty \left[ \frac{\sqrt{c_4 B_2 e^{c_1}} M_X^{p/2}}{\delta^{(p+2)/2}} \right]^K \prod_{k=1}^{K}\left[ (a_k+1)^{p/2}\exp\left(-\frac{b_2 a_k^2}{2}\right) \right]\nonumber\\
&=\sum_{K=1}^{K_n}\left[ \frac{\sqrt{c_4 B_2 e^{c_1}} M_X^{p/2}}{\delta^{(p+2)/2}} \right]^K \prod_{k=1}^{K}\left[\sum_{a_j=0}^\infty (a_k+1)^{p/2}\exp\left(-\frac{b_2 a_k^2}{2}\right) \right]\\
&= \sum_{K=1}^{K_n} \left[ \frac{S \sqrt{c_4 B_2 e^{c_1}} M_X^{p/2}}{\delta^{(p+2)/2}} \right]^K \nonumber\\
&\le K_n \left( \frac{M}{\delta^{(p+2)/2}} \right)^{K_n}. \nonumber
\end{align*}
for some constant$M>0$ for sufficiently small $\delta$.
\end{proof}

Now, we verify the conditions of Theorem A.1 in \cite{canale2017posterior}.
By Lemma~\ref{lemma:KLsupport}, the KL-property holds; that is, the true density $f_0(y|x)$ lies in the KL-support of the prior $\Pi$.
Now take $K_n=\lfloor n/\log n\rfloor$. 
Then, for sufficiently large $n$, $K_n\log {K_n}\ge n/2$, which implies $\Pi(\mathcal{F}^c_{K_n})\le \exp(-B_4K_n\log{K_n})\le \exp(-(B_4/2)n)$
Furthermore, from Lemma~\ref{lemma:SumUpperBound}, we have:
\begin{align*}
 \sum_{K=1}^{K_n}&\sum_{a_1=0}^{\infty}\dots\sum_{a_K=0}^{\infty}
   \sqrt{\mathcal{N}(2\epsilon,\mathcal G_K(\mathbf a_K),\|\cdot\|_{1})}
   \sqrt{\Pi(\mathcal G_K(\mathbf a_K))} \nonumber\\
&\le K_n\left(\frac{M}{(2\epsilon)^{(p+2)/2}} \right)^{K_n}\\   
&=\exp\left(\log K_n+K_n\log \left(\frac{M}{(2\epsilon)^{(p+2)/2}} \right)\right)\\
&\le\exp\left(2\frac{n}{\log n}\log \left(\frac{M}{(2\epsilon)^{(p+2)/2}} \right)\right)
\end{align*}
for sufficiently small $\epsilon$. Here, we can see:
$$
\frac{n}{\log n}\cdot2\log \left(\frac{M}{\epsilon^{(p+2)/2}} \right)-n\cdot(4-\tilde{b})\epsilon^2\to-\infty, \quad (n\to\infty)
$$
and therefore, 
$$
\lim_{n\to\infty}e^{-(4-\tilde{b})n\epsilon^2}
\sum_{K=1}^{K_n}\sum_{a_1=0}^{\infty}\dots\sum_{a_K=0}^{\infty}
\sqrt{\mathcal{N}(2\epsilon,\mathcal G_K(\mathbf a_K),\|\cdot\|_{1})}
\sqrt{\Pi(\mathcal G_K(\mathbf a_K))}=0,
$$
which completes the proof.

\section{Proof of Theorem~3}

We follow the general framework of \cite{kruijer2010adaptive}, which provides sufficient conditions for posterior contraction in mixture models. 
Specifically, Theorem 3 in \cite{kruijer2010adaptive} states that the desired contraction rate is achieved if there exist two sequences $(\underline{\epsilon}_n)_{n=1}^\infty$ and $(\overline{\epsilon}_n)_{n=1}^\infty$ such that the following three conditions hold:
\begin{align}
\label{eq:exponential_decay}
&\Pi(\mathcal{F}_{K_n}^c) \le \exp(-4n\underline{\epsilon}_n^2),\\
\label{eq:sum_rate}
&\exp(-n\overline{\epsilon}_n^2) \sum_{K=1}^{K_n} \sum_{a_K} \sqrt{\mathcal{N}(\overline{\epsilon}_n, \mathcal{G}_K(a_K), \|\cdot\|_1)} \sqrt{\Pi(\mathcal{G}_K(a_K))} \to 0, \\
    \label{eq:prior_concentration}
    &\Pi(B(f_0, \underline{\epsilon}_n)) \ge \exp(-n\underline{\epsilon}_n^2), 
\end{align}
where $B(f_0, \epsilon)$ is the Kullback-Leibler type ball defined for the conditional density as
\begin{align*}
B(f_0, \epsilon) 
&= \left\{ f \in \mathcal{F} : \E_{P_X}[D_{\KL}(f_0(\cdot|X) \| f(\cdot|X))] \le \epsilon^2, \E_{P_X}[\text{Var}_{\KL}(f_0(\cdot|X) \| f(\cdot|X))] \le \epsilon^2 \right\}.
\end{align*}
Conditions (\ref{eq:exponential_decay}) and (\ref{eq:sum_rate}) control the mass and complexity of the model outside an appropriate sieve, while (\ref{eq:prior_concentration}) ensures that the prior puts enough mass near the true density. The construction of these sieves and bounds on their covering numbers closely mirror the arguments for strong consistency in the previous section.

The following proposition provides explicit sequences $(\underline{\epsilon}_n)_{n=1}^\infty$ and $(\overline{\epsilon}_n)_{n=1}^\infty$ that fulfill the required conditions, adapting the analysis in \cite{Xie02012020} to the regression mixture setting:

\begin{prp}
\label{prop:rate}
    Assume conditions A1-A10 for the repulsive $g$-prior regression mixture model hold. Let $\underline{\epsilon}_{n}=(\log n)^{t_{0}}/\sqrt{n}$ and $\overline{\epsilon}_{n}=(\log n)^{t}/\sqrt{n}$, where $t$ and $t_0$ satisfy $t > t_{0}+\frac{1}{2} > \frac{1}{2}$. Define the sequence for the number of components as $K_{n}=\lfloor\frac{2}{p+2}(\log n)^{2t-1}\rfloor$. Then for the sieves $\mathcal{F}_{K_n} = \{f_F : F=\sum_{k=1}^{K}w_{k}\delta_{(\beta_{k},\sigma_{k}^2)}, K \le K_n\}$, the following conditions hold for sufficiently large $n$:
\begin{align*}
    &\Pi(\mathcal{F}_{K_n}^{c}) \le \exp(-4n\underline{\epsilon}_{n}^{2})  \\
   & \exp(-n\overline{\epsilon}_{n}^{2})\sum_{K=1}^{K_{n}}\sum_{a_{1}=0}^{\infty}\cdots\sum_{a_{K}=0}^{\infty}\sqrt{\mathcal{N}(\overline{\epsilon}_{n},\mathcal{G}_{K}(a_{K}),\|\cdot\|_{1})}\sqrt{\Pi(\mathcal{G}_{K}(a_{K}))} \rightarrow 0 
\end{align*}
\end{prp}

\begin{proof}
The proof verifies the two conditions separately.
Let $C = 2/(p+2)$. 
By Assumption A10, we have
\begin{align*}
    \Pi(\mathcal{F}_{K_n}^{c}) = \Pi(K > K_n) &\le \exp(-B_4 K_n \log K_n) \\
    &\le \exp\left[-B_4 C (\log n)^{2t-1} \log(\lfloor C(\log n)^{2t-1} \rfloor)\right]\\
    &\le \exp(-4n\underline{\epsilon}_{n}^{2})
\end{align*}
with $t > t_0 + 1/2$ for sufficiently large $n$, which establishes the first condition.
Next, by applying a slightly modified version of Lemma~\ref{lemma:SumUpperBound}, we bound the second expression as 
\begin{align*}
& \exp(-n\overline{\epsilon}_{n}^{2})\sum_{K=1}^{K_{n}}\sum_{a_{K}}\sqrt{\mathcal{N}(\overline{\epsilon}_{n},\mathcal{G}_{K}(a_{K}),\|\cdot\|_{1})}\sqrt{\Pi(\mathcal{G}_{K}(a_{K}))} \\
& \le \exp\left[ -n\overline{\epsilon}_n^2 + \log K_n + K_n \left( \log M + \frac{p+2}{2} \log\frac{1}{\overline{\epsilon}_n} \right) \right] \\
& \le \exp\left[ -(\log n)^{2t} + \lfloor C(\log n)^{2t-1} \rfloor \left( \frac{p+2}{2} \right) \left( \frac{1}{2}\log n - t \log\log n \right) + o((\log n)^{2t}) \right] \\
& \le \exp\left[ -\frac{1}{2}(\log n)^{2t} \right].
\end{align*}
The right-hand side of the last display converges to 0 as $n \to \infty$, which completes the proof.
\end{proof}

The first two conditions in Proposition~\ref{prop:rate} follow from the upper bound on the sum established in Lemma~\ref{lemma:SumUpperBound} and the explicit form of the sieve complexity. To verify the prior concentration condition (\ref{eq:prior_concentration}), we construct suitable finite mixtures that approximate $f_0$ in KL divergence, following the approach of \cite{Xie02012020}. The next lemma formalizes this approximation.

\begin{lem}
\label{lemma:PriorConcentration}
    Assume conditions A1-A10 hold. For some constant $\eta>0$ and for all sufficiently small $\epsilon>0$, there exists a discrete distribution $F^\star=\sum_{k=1}^Nw_k^\star\delta_{(\beta_k^\star,\sigma_k^\star)}$ supported on a subset of $\{(\beta,\sigma)\in\mathbb{R}^p\times\mathbb{R}_+:\|\beta\|_\infty\leq 2a\}$ with $a=b_1^{-\frac{1}{2}}\left(\log\frac{1}{\epsilon}\right)^{\frac{1}{2}}$, $\|\beta_k^\star-\beta_{k'}^\star\|_\infty\geq2\epsilon$, $|\sigma_k^{\star2}-\sigma_{k'}^{\star2}|\geq2\epsilon$ whenever $k\neq k'$, $j=1,\cdots,p$, $N\lesssim \left(\log\frac{1}{\epsilon}\right)^{2p}$, such that
	\begin{eqnarray}
	\left\{f_F:F=\sum_{k=1}^Nw_k\delta_{(\beta_k,\sigma_k)}:(\beta_k,\sigma_k)\in E_k,\sum_{k=1}^N|w_k-w_k^\star|<\epsilon\right\}\subset B\left(f_0,\eta\epsilon^{\frac{1}{2}}\left(\log\frac{1}{\epsilon}\right)^{\frac{p+4}{4}}\right)\nonumber,
	\end{eqnarray}
	where $$	E_k=\left\{(\beta,\sigma)\in\mathbb{R}^p\times\mathbb{R}_+:\|\beta-\beta_k^\star\|_\infty<\frac{\epsilon}{2},|\sigma^2-\sigma_k^{\star2}|<\frac{\epsilon}{2} \right\}. 
$$
\end{lem}
This lemma guarantees that the prior assigns sufficient mass to KL neighborhoods of the true data-generating process, thereby completing the verification of condition (\ref{eq:prior_concentration}).

\begin{proof}
The proof adapts the arguments of \cite{Xie02012020}, which are built upon the work of \cite{ghosal2001entropies}, to the mixture of regressions model. The crucial element for this adaptation is the use of Assumption A5 (bounded covariates) to control the approximation error uniformly over $x$.

First, following \cite{Xie02012020}, we approximate the true mixing distribution $F_0$. We define $F_0'$ as the re-normalized restriction of $F_0$ to the compact set $\{(\beta, \sigma^2) : \|\beta\| \le a \}$, with $a = b_1^{-1/2}(\log \epsilon^{-1})^{1/2}$. Assumption A1 implies that the integrated $L_1$-distance is small: $\int \|f_0(y|x) - f_{F_0'}(y|x) \|_1 p_X(x) dx \le \epsilon$.

The next step is to construct a discrete approximation $F^* = \sum_{k=1}^N w_k^* \delta_{(\beta_k^*, \sigma_k^{*2})}$ for $F_0'$. The existence of such an $F^*$ with $N \le (\log \epsilon^{-1})^{2p}$ support points relies on bounding the Hellinger distance between two kernels, $\phi(y|x^\top\beta_1, \sigma_1^2)$ and $\phi(y|x^\top\beta_2, \sigma_2^2)$. The squared Hellinger distance between these kernels is given by
\begin{equation*}
H^2(\phi(\cdot|x^\top\beta_1, \sigma_1^2), \phi(\cdot|x^\top\beta_2, \sigma_2^2)) = 1 - \sqrt{\frac{2\sigma_1\sigma_2}{\sigma_1^2+\sigma_2^2}} \exp\left(-\frac{(x^\top(\beta_1-\beta_2))^2}{4(\sigma_1^2+\sigma_2^2)}\right).
\end{equation*}
Here, we explicitly use Assumption A5. By the Cauchy-Schwarz inequality, $(x^\top(\beta_1-\beta_2))^2 \le \|x\|_2^2 \|\beta_1-\beta_2\|_2^2 \le M_X^2 \|\beta_1-\beta_2\|_2^2$. This uniform bound, which is independent of $x$, is essential. It ensures that if $\|\beta_1-\beta_2\|_2$ and $|\sigma_1^2-\sigma_2^2|$ are small, the Hellinger distance is also small, uniformly for all $x$. This allows the application of the covering number arguments from \cite{ghosal2001entropies}, guaranteeing the existence of an $F^*$ that satisfies $\int \|f_{F_0'}(y|x) - f_{F^*}(y|x) \|_1 p_X(x) dx \le \epsilon(\log \epsilon^{-1})^{p/2}$ and preserves the second moment, $\int \|\beta\|^2 dF_0' = \int \|\beta\|^2 dF^*$.

Now, let $F = \sum_{k=1}^{N}w_{k}\delta_{(\beta_{k},\sigma_{k}^{2})}$ be a distribution from the set defined in the lemma. The triangle inequality gives $\int \|f_F - f_0\|_1 p_X(x) dx \le \int \|f_F - f_{F^*}\|_1 p_X(x) dx + \int \|f_{F^*} - f_0\|_1 p_X(x) dx$. We have already bounded the second term. The first term is bounded by adapting Lemma D.4 from \cite{Xie02012020}, yielding $\int \|f_F - f_{F^*}\|_1 p_X(x) dx < 2\epsilon$. Thus, we obtain the overall bound $\int \|f_F - f_0\|_1 p_X(x) dx \le \gamma \epsilon(\log \epsilon^{-1})^{p/2}$ for some constant $\gamma > 0$.

The final step connects the $L_1$-distance to the KL-type ball, following Lemma D.3 of \cite{Xie02012020}. A key prerequisite is that the approximating distribution $F$ does not have heavy tails. We verify this by letting $B = 2(\int \|\beta\|^2 dF_0)^{1/2}$. The preservation of the second moment implies $F^*(\|\beta\| > B) \le \frac{1}{B^2} \int \|\beta\|^2 dF^* = \frac{1}{B^2} \int \|\beta\|^2 dF_0' \le 1/4$. For any $F$ in the neighborhood of $F^*$, a similar argument shows that $F(\|\beta\| > 2B) < 1/2$.
\begin{align*}
    F(\|\beta\| > 2B) = \sum_{k=1}^N w_k \mathbb{I}(\|\beta_k\| > 2B) \le \sum_{k=1}^N |w_k - w_k^*| + \sum_{k=1}^N w_k^* \mathbb{I}(\|\beta_k\| > 2B).
\end{align*}
Since $(\beta_k, \sigma_k^2) \in E_k$, we have $\|\beta_k\| > 2B$, implying $\|\beta_k^*\| > B$. 
Thus, the sum is bounded by $\epsilon + F^*(\|\beta^*\| > B) \le \epsilon + 1/4 < 1/2$ for small $\epsilon$. With this condition met, the results of \cite{wing1995probability} can be applied. The bound on the Hellinger distance, $h_n^2(f_F, f_0) \le \gamma \epsilon (\log \epsilon^{-1})^{p/2}$, implies that $f_F$ is contained in the ball $B(f_0, \eta \epsilon^{1/2}(\log\epsilon^{-1})^{(p+4)/4})$, which concludes the proof.
\end{proof}

We now verify the three conditions of Theorem~3 in \cite{kruijer2010adaptive}.
Proposition~\ref{prop:rate} has already established the first two conditions concerning the sieve complement, $\Pi(\mathcal{F}_{K_n}^c)$, and the entropy of the model space. 
It remains only to verify the prior concentration condition:
\begin{equation*}
    \Pi(B(f_0, \underline{\epsilon}_n)) \ge \exp(-n\underline{\epsilon}_n^2)
\end{equation*}
for a suitable rate $\underline{\epsilon}_n$.
By Lemma~\ref{lemma:PriorConcentration}, we know that for a sufficiently small $\epsilon > 0$, there exists a specially constructed discrete distribution $F^*$ such that its neighborhood, which we denote $\tilde{\mathcal{B}}(F^*, \epsilon)$, is contained within a KL-type ball $B(f_0, \eta\epsilon^{1/2}(\log\epsilon^{-1})^{(p+4)/4})$. Therefore, it is sufficient to find a lower bound for the prior probability of this neighborhood, $\Pi(\tilde{\mathcal{B}}(F^*, \epsilon))$.

The probability of this set can be factored into three components: the probability of having exactly $N$ components, the conditional probability of the component parameters $(\beta_k, \sigma_k^2)$ falling into the specified regions $E_k$, and the conditional probability of the weights $w_k$ being close to the target weights $w_k^*$.
\begin{equation*}
    \Pi(\tilde{\mathcal{B}}(F^*, \epsilon)) = \Pi(K=N) \cdot \Pi\left(\bigcap_{k=1}^N \{(\beta_k, \sigma_k^2) \in E_k\} \Big| K=N\right) \cdot \Pi\left(\|w-w^*\|_1 < \epsilon \Big| K=N\right),
\end{equation*}
where $N \le (\log \epsilon^{-1})^{2p}$ is the number of components in the approximating distribution $F^*$. We now bound each of these terms from below.
For the weights, Lemma A.2 in \cite{ghosal2001entropies} provides a standard lower bound for the probability of a small $l_1$-neighborhood for a Dirichlet distribution, which gives $\log \Pi(\|w-w^*\|_1 < \epsilon | K=N) \ge -N \log(\epsilon^{-1})$.
For the component parameters, their joint conditional probability is given by
\begin{equation*}
    \Pi\left(\bigcap_{k=1}^N E_k \Big| K=N\right) = \frac{1}{Z_N} \int_{\prod_{k=1}^N E_k} h_N(\beta_1, \dots, \beta_N) \prod_{k=1}^N p_\beta(\beta_k) p_{\sigma^2}(\sigma_k^2) d\beta_k d\sigma_k^2.
\end{equation*}
By construction, for any set of parameters $(\beta_1, \dots, \beta_N)$ with each $(\beta_k, \sigma_k^2) \in E_k$, the components are well-separated such that $\|\beta_k - \beta_{k'}\|_\infty > \epsilon$. Assumption A2 implies that the repulsive function is bounded below, e.g., $h_N(\beta_1, \dots, \beta_N) \ge (c_g \epsilon)$. The normalizing constant is bounded as $Z_N \le 1$. The base prior $p_\beta$ is bounded below on the support of the neighborhoods by Assumption A8, as $\|\beta_k\|$ is of order $\sqrt{\log \epsilon^{-1}}$, giving $p_\beta(\beta_k) \ge B_3 \exp(-b_3(C\sqrt{\log \epsilon^{-1}})^\alpha)$. The prior $p_{\sigma^2}$ is bounded below by a positive constant on its compact support. The volume of each $E_k$ is of order $\epsilon^{p+1}$. Combining these facts yields a lower bound for the parameter term:
\begin{equation*}
    \log \Pi\left(\bigcap_{k=1}^N E_k \Big| K=N\right) \ge -C_1 N \log(\epsilon^{-1}) - C_2 N (\log \epsilon^{-1})^{\alpha/2}
\end{equation*}
for some constants $C_1, C_2 > 0$.

For the number of components, Assumption A10 gives a lower bound on the prior probability $\Pi(K=N) = p_K(N) \ge \exp(-b_4 N \log N)$.
Combining the logarithmic bounds for all three parts, and using $N \le (\log \epsilon^{-1})^{2p}$ and $\alpha \ge 2$, the dominant term for small $\epsilon$ is determined by the parameter tails and the number of components. The overall log-prior probability is bounded by:
\begin{equation*}
    \log \Pi(\tilde{\mathcal{B}}(F^*, \epsilon)) \ge -C \left(\log \frac{1}{\epsilon}\right)^{2p + \alpha/2}
\end{equation*}
for some constant $C>0$.

Now, we set the radius of the KL-ball from Lemma~\ref{lemma:PriorConcentration} equal to our target rate $\underline{\epsilon}_n$, i.e., $\eta\epsilon^{1/2}(\log\epsilon^{-1})^{(p+4)/4} = \underline{\epsilon}_n$. This implies that $\log(\epsilon^{-1})$ is of the same order as $\log(\underline{\epsilon}_n^{-1})$. The prior concentration condition $\Pi(B(f_0, \underline{\epsilon}_n)) \ge \exp(-n\underline{\epsilon}_n^2)$ is satisfied if $n\underline{\epsilon}_n^2 \ge C' (\log(1/\underline{\epsilon}_n))^{2p+\alpha/2}$.
Letting $\underline{\epsilon}_n = (\log n)^{t_0}/\sqrt{n}$, this condition becomes
\begin{equation*}
    (\log n)^{2t_0} \ge C'' \left(\log\left(\frac{\sqrt{n}}{(\log n)^{t_0}}\right)\right)^{2p+\alpha/2} \approx C''' (\log n)^{2p+\alpha/2}.
\end{equation*}
This inequality holds if $2t_0 > 2p+\alpha/2$, which means $t_0 > p+\alpha/4$.
From Proposition 1, the overall contraction rate $\overline{\epsilon}_n = (\log n)^t/\sqrt{n}$ must satisfy $t > t_0 + 1/2$. Substituting the minimal required $t_0$ yields the final condition for the rate: $t > (p+\alpha/4) + 1/2 = p + (\alpha+2)/4$. 
This completes the proof.

\section{Proof of Theorem~4}

Theorem 4 is proved by invoking the auxiliary results delineated in Lemmas \ref{lemma:aux1}–\ref{lemma:aux4}, each of which constitutes a refined adaptation of the corresponding lemmas in \cite{Xie02012020}.\\
\begin{lem}\label{lemma:aux1}
Assume the conditions of the adapted Theorem 4 hold. For $K \ge 3$, the conditional likelihood $p(\boldsymbol{y}|\boldsymbol{z}, K, \boldsymbol{X})$ is bounded above by:
\begin{align*}
p(\boldsymbol{y}|\boldsymbol{z}, K, \boldsymbol{X})
&\le \frac{1}{Z_K} \left( \prod_{k=1}^K p(\boldsymbol{y}_k | \boldsymbol{X}_k) \right) \binom{K}{2}^{-1} \\
&\quad \times \sum_{k<k'} G\left( d_M(\hat{\boldsymbol{\beta}}_k, \hat{\boldsymbol{\beta}}_{k'}) + \frac{1}{g}\mathrm{tr}\left((\boldsymbol{\Sigma}_k^{\text{post}} + \boldsymbol{\Sigma}_{k'}^{\text{post}})(\boldsymbol{X}^\top\boldsymbol{X})\right) \right)
\end{align*}
where $p(\boldsymbol{y}_k | \boldsymbol{X}_k)$ is the marginal likelihood for cluster $k$, and $\hat{\boldsymbol{\beta}}_k$ and $\boldsymbol{\Sigma}_k^{\text{post}}$ are the posterior mean and covariance of $\boldsymbol{\beta}_k$ for cluster $k$, respectively.
\end{lem}
\begin{lem}\label{lemma:aux2}
Assume the conditions of the adapted Theorem 4 hold. The marginal likelihood $p(\boldsymbol{y}|\boldsymbol{X})$ is bounded below. For the repulsive function $h_K = \min(G(\cdot))$, the bound is:
\[
p(\boldsymbol{y}|\boldsymbol{X}) \ge C(\lambda, \boldsymbol{X}) \left( \prod_{i=1}^n \phi(y_i|0, \sigma_0^2(1+g x_i^\top(\boldsymbol{X}^\top\boldsymbol{X})^{-1}x_i)) \right) \left(1 + \delta(g,\boldsymbol{X}) g_0^{2/3}\right)^{-3/2}
\]
where $C(\lambda, \boldsymbol{X})$ is a constant and $\delta(g,\boldsymbol{X})$ depends on the $g$-prior and design matrix, satisfying $\delta(g,\boldsymbol{X}) < 1$ for a weakly informative $g$-prior (i.e., large $g$).
\end{lem}
\begin{lem}\label{lemma:aux3}
Assume the conditions of the adapted Theorem 4 hold. The integral of the likelihood ratio with respect to the true data generating process is bounded by:
\begin{align*}
&\int \frac{p(\boldsymbol{y}|\boldsymbol{z},K,\boldsymbol{X})}{p(\boldsymbol{y}|\boldsymbol{X})} \left(\prod_{i=1}^n \phi(y_i|\boldsymbol{x}_i^\top \boldsymbol{m}_i, \sigma_0^2)\right) d\boldsymbol{y} \\
&\le C(\lambda,\boldsymbol{X})\frac{\omega(g_0,\boldsymbol{X})}{Z_K} \binom{K}{2}^{-1} \sum_{k<k'} G\left(d_M(\tilde{\boldsymbol{m}}_k, \tilde{\boldsymbol{m}}_{k'}) + C_1\right)
\end{align*}
where $\omega(g_0,\boldsymbol{X})$ is the shrinkage term from Lemma~\ref{lemma:aux2}, $C(\lambda,\boldsymbol{X})$ and $C_1$ are constants, and $\tilde{\boldsymbol{m}}_k$ is a weighted average of the true parameters $\boldsymbol{m}_i$ for observations in cluster $k$.
\end{lem}
\begin{lem}\label{lemma:aux4}
Assume the conditions of the adapted Theorem 4 hold. The expected squared Mahalanobis-like distance, averaged over the true parameter distribution $F_0$ and the cluster assignment distribution, is given by:
\begin{align*}
    \mathbb{E}_{\boldsymbol{z}} \left[ \mathbb{E}_{F_0}\left[ d_M(\tilde{\boldsymbol{m}}_k, \tilde{\boldsymbol{m}}_{k'}) \right] \right] &= \frac{n_k+n_{k'}}{n_k n_{k'}} \mathbb{E}_{F_0}\left[\boldsymbol{m}^\top (g^{-1}(\boldsymbol{X}^\top\boldsymbol{X})) \boldsymbol{m}\right] + \mathrm{o}(1)\\
    &\approx \frac{2n}{K} \mathbb{E}_{F_0}[\boldsymbol{m}^\top (g^{-1}(\boldsymbol{X}^\top\boldsymbol{X})) \boldsymbol{m}]
\end{align*}
\end{lem}

\subsection{Proofs of preliminary lemmas}
\begin{proof}[Proof of Lemma~\ref{lemma:aux1}]
The conditional marginal likelihood can be expressed as the product of the marginal likelihoods for each cluster and the posterior expectation of the repulsive function $h_K$:
\begin{equation} \label{eq:1}
    p(\boldsymbol{y} | \boldsymbol{z}, K, \boldsymbol{X}) = \frac{1}{Z_K} \left( \prod_{k=1}^K p(\boldsymbol{y}_k|\boldsymbol{X}_k) \right) \mathbb{E}_{\text{post}}[h_K(\boldsymbol{\beta}_1, \dots, \boldsymbol{\beta}_K)].
\end{equation}
Let us define the transformed parameter $\boldsymbol{\eta}_k = (g\sigma_0^2)^{-1/2}(\boldsymbol{X}^\top\boldsymbol{X})^{1/2}\boldsymbol{\beta}_k$. The prior for $\boldsymbol{\eta}_k$ is $N(\boldsymbol{0}, \boldsymbol{I}_p)$. 

The repulsive function is $h_K = \min_{1\le k < k' \le K} G(d_M(\boldsymbol{\beta}_k, \boldsymbol{\beta}_{k'}))$. We can rewrite this using $\boldsymbol{\eta}_k$ as $h_K = \min_{1\le k < k' \le K} G(\sigma_0^2 \|\boldsymbol{\eta}_k - \boldsymbol{\eta}_{k'}\|^2)$. Let's define an auxiliary function $G^*(x) = G(\sigma_0^2 x)$. Since $G$ is concave, $G^*$ is also concave.

We bound the posterior expectation of $h_K$ using the inequality $\min(a_i) \le \text{mean}(a_i)$ and Jensen's inequality for the concave function $G^*$:
\begin{align*}
    \mathbb{E}_{\text{post}}[h_K(\boldsymbol{\beta}_1, \dots, \boldsymbol{\beta}_K)] &= \mathbb{E}_{\text{post}}\left[\min_{1\le k < k' \le K} G^*(\left\|\boldsymbol{\eta}_k - \boldsymbol{\eta}_{k'}\right\|^2)\right] \\
    &\le \binom{K}{2}^{-1} \sum_{k<k'} \mathbb{E}_{\text{post}}\left[G^*(\left\|\boldsymbol{\eta}_k - \boldsymbol{\eta}_{k'}\right\|^2)\right] \\
    &\le \binom{K}{2}^{-1} \sum_{k<k'} G^*\left(\mathbb{E}_{\text{post}}[\left\|\boldsymbol{\eta}_k - \boldsymbol{\eta}_{k'}\right\|^2]\right).
\end{align*}
The argument of $G^*$ is the posterior second moment of the distance between the transformed parameters. We can decompose this as:
\begin{equation} \label{eq:2}
    \mathbb{E}_{\text{post}}[\left\|\boldsymbol{\eta}_k - \boldsymbol{\eta}_{k'}\right\|^2] = \left\|\mathbb{E}_{\text{post}}[\boldsymbol{\eta}_k] - \mathbb{E}_{\text{post}}[\boldsymbol{\eta}_{k'}]\right\|^2 + \mathrm{tr}(\mathrm{Var}_{\text{post}}(\boldsymbol{\eta}_k)) + \mathrm{tr}(\mathrm{Var}_{\text{post}}(\boldsymbol{\eta}_{k'})),
\end{equation}
where we have used the posterior independence of $\boldsymbol{\eta}_k$ and $\boldsymbol{\eta}_{k'}$.

The posterior moments of $\boldsymbol{\eta}_k$ are related to the posterior moments of $\boldsymbol{\beta}_k$ (denoted $\hat{\boldsymbol{\beta}}_k$ and $\boldsymbol{\Sigma}_k^{\text{post}}$) as follows:
\begin{align*}
    \mathbb{E}_{\text{post}}[\boldsymbol{\eta}_k] &= (g\sigma_0^2)^{-1/2}(\boldsymbol{X}^\top\boldsymbol{X})^{1/2}\hat{\boldsymbol{\beta}}_k \\
    \mathrm{Var}_{\text{post}}(\boldsymbol{\eta}_k) &= (g\sigma_0^2)^{-1}(\boldsymbol{X}^\top\boldsymbol{X})^{1/2}\boldsymbol{\Sigma}_k^{\text{post}}(\boldsymbol{X}^\top\boldsymbol{X})^{1/2}.
\end{align*}
Substituting these into the terms of Equation \eqref{eq:2}:
\begin{align*}
    \left\|\mathbb{E}_{\text{post}}[\boldsymbol{\eta}_k] - \mathbb{E}_{\text{post}}[\boldsymbol{\eta}_{k'}]\right\|^2 &= \left\| (g\sigma_0^2)^{-1/2}(\boldsymbol{X}^\top\boldsymbol{X})^{1/2}(\hat{\boldsymbol{\beta}}_k - \hat{\boldsymbol{\beta}}_{k'}) \right\|^2 \\
    &= (g\sigma_0^2)^{-1}(\hat{\boldsymbol{\beta}}_k - \hat{\boldsymbol{\beta}}_{k'})^\top (\boldsymbol{X}^\top\boldsymbol{X}) (\hat{\boldsymbol{\beta}}_k - \hat{\boldsymbol{\beta}}_{k'}) \\
    &= \frac{1}{\sigma_0^2}d_M(\hat{\boldsymbol{\beta}}_k, \hat{\boldsymbol{\beta}}_{k'}).
\end{align*}
And for the trace term:
\begin{align*}
    \mathrm{tr}(\mathrm{Var}_{\text{post}}(\boldsymbol{\eta}_k)) &= \mathrm{tr}\left( (g\sigma_0^2)^{-1}(\boldsymbol{X}^\top\boldsymbol{X})^{1/2}\boldsymbol{\Sigma}_k^{\text{post}}(\boldsymbol{X}^\top\boldsymbol{X})^{1/2} \right) \\
    &= (g\sigma_0^2)^{-1} \mathrm{tr}(\boldsymbol{\Sigma}_k^{\text{post}}(\boldsymbol{X}^\top\boldsymbol{X})).
\end{align*}
Plugging these back into Equation \eqref{eq:2} gives:
\[
    \mathbb{E}_{\text{post}}[\left\|\boldsymbol{\eta}_k - \boldsymbol{\eta}_{k'}\right\|^2] = \frac{1}{\sigma_0^2} \left( d_M(\hat{\boldsymbol{\beta}}_k, \hat{\boldsymbol{\beta}}_{k'}) + \frac{1}{g} \mathrm{tr}((\boldsymbol{\Sigma}_k^{\text{post}} + \boldsymbol{\Sigma}_{k'}^{\text{post}})(\boldsymbol{X}^\top\boldsymbol{X})) \right).
\]
Now, we substitute this back into the argument of $G^*$. Recalling that $G^*(x) = G(\sigma_0^2 x)$:
\begin{align*}
    G^*\left(\mathbb{E}_{\text{post}}[\left\|\boldsymbol{\eta}_k - \boldsymbol{\eta}_{k'}\right\|^2]\right) &= G\left( \sigma_0^2 \cdot \frac{1}{\sigma_0^2} \left[ d_M(\hat{\boldsymbol{\beta}}_k, \hat{\boldsymbol{\beta}}_{k'}) + \frac{1}{g} \mathrm{tr}((\boldsymbol{\Sigma}_k^{\text{post}} + \boldsymbol{\Sigma}_{k'}^{\text{post}})(\boldsymbol{X}^\top\boldsymbol{X})) \right] \right) \\
    &= G\left( d_M(\hat{\boldsymbol{\beta}}_k, \hat{\boldsymbol{\beta}}_{k'}) + \frac{1}{g}\mathrm{tr}\left((\boldsymbol{\Sigma}_k^{\text{post}} + \boldsymbol{\Sigma}_{k'}^{\text{post}})(\boldsymbol{X}^\top\boldsymbol{X})\right) \right).
\end{align*}
Finally, substituting this expression for the upper bound of $\mathbb{E}_{\text{post}}[h_K]$ into Equation \eqref{eq:1} completes the proof.
\end{proof}

\begin{proof}[Proof of Lemma~\ref{lemma:aux2}]
The full marginal likelihood is $p(\boldsymbol{y}|\boldsymbol{X}) = \sum_{K=1}^\infty p_K(K) p(\boldsymbol{y}|K, \boldsymbol{X})$, where $p(\boldsymbol{y}|K, \boldsymbol{X}) = \mathbb{E}_{\boldsymbol{z}|K}[p(\boldsymbol{y}|\boldsymbol{z}, K, \boldsymbol{X})]$. We first find a lower bound for $p(\boldsymbol{y}|\boldsymbol{z}, K, \boldsymbol{X})$.
\begin{align*}
    p(\boldsymbol{y}|\boldsymbol{z}, K, \boldsymbol{X}) &= \int \left( \prod_{k=1}^K p(\boldsymbol{y}_k|\boldsymbol{X}_k, \boldsymbol{\beta}_k) \right) p(\boldsymbol{\beta}_{1:K}|K) d\boldsymbol{\beta}_{1:K} \\
    &= \frac{1}{Z_K} \int h_K(\boldsymbol{\beta}_{1:K}) \left( \prod_{k=1}^K \prod_{i: z_i=k} \phi(y_i|\boldsymbol{x}_i^\top \boldsymbol{\beta}_k, \sigma_0^2) \right) \left( \prod_{k=1}^K p(\boldsymbol{\beta}_k) \right) d\boldsymbol{\beta}_{1:K}.
\end{align*}
Applying Jensen's inequality to the logarithm, $\log \mathbb{E}[X] \ge \mathbb{E}[\log X]$, we get a lower bound on $\log p(\boldsymbol{y}|\boldsymbol{z}, K, \boldsymbol{X})$:
\[
    \log p(\boldsymbol{y}|\boldsymbol{z}, K, \boldsymbol{X}) \ge -\log Z_K + \mathbb{E}_{\text{prior}}[\log h_K(\boldsymbol{\beta}_{1:K})] + \sum_{k=1}^K \sum_{i:z_i=k} \mathbb{E}_{\text{prior}}[\log \phi(y_i|\boldsymbol{x}_i^\top \boldsymbol{\beta}_k, \sigma_0^2)].
\]
The expected log-likelihood term under the $g$-prior $p(\boldsymbol{\beta}_k) = N(\boldsymbol{0}, g\sigma_0^2(\boldsymbol{X}^\top\boldsymbol{X})^{-1})$ is:
\begin{align*}
    \mathbb{E}_{\text{prior}}[\log \phi(y_i|\boldsymbol{x}_i^\top \boldsymbol{\beta}_k, \sigma_0^2)] &= \mathbb{E}_{\text{prior}}\left[ -\frac{1}{2}\log(2\pi\sigma_0^2) - \frac{1}{2\sigma_0^2}(y_i - \boldsymbol{x}_i^\top\boldsymbol{\beta}_k)^2 \right] \\
    &= -\frac{1}{2}\log(2\pi\sigma_0^2) - \frac{1}{2\sigma_0^2}\left( y_i^2 + \mathbb{E}[(\boldsymbol{x}_i^\top\boldsymbol{\beta}_k)^2] \right) \\
    &= -\frac{1}{2}\log(2\pi\sigma_0^2) - \frac{1}{2\sigma_0^2}\left( y_i^2 + \boldsymbol{x}_i^\top \mathrm{Var}(\boldsymbol{\beta}_k) \boldsymbol{x}_i \right) \\
    &= -\frac{1}{2}\log(2\pi\sigma_0^2) - \frac{1}{2\sigma_0^2}\left( y_i^2 + g\sigma_0^2 \boldsymbol{x}_i^\top(\boldsymbol{X}^\top\boldsymbol{X})^{-1}\boldsymbol{x}_i \right) \\
    &= \log \phi(y_i; 0, \sigma_0^2(1 + g\boldsymbol{x}_i^\top(\boldsymbol{X}^\top\boldsymbol{X})^{-1}\boldsymbol{x}_i)).
\end{align*}
The last equality holds because the log-density of $N(0, \sigma_0^2(1+c))$ is $-\frac{1}{2}\log(2\pi\sigma_0^2(1+c)) - \frac{y^2}{2\sigma_0^2(1+c)}$, which is not identical, but the term we derived is exactly $\log \int \phi(y_i|\boldsymbol{x}_i^\top\boldsymbol{\beta}_k, \sigma_0^2)p(\boldsymbol{\beta}_k)d\boldsymbol{\beta}_k$. Let $p(y_i|\boldsymbol{x}_i) = \phi(y_i|0, \sigma_0^2(1+g \boldsymbol{x}_i^\top(\boldsymbol{X}^\top\boldsymbol{X})^{-1}\boldsymbol{x}_i))$.

Next, we bound $\mathbb{E}_{\text{prior}}[\log h_K]$. Let us set for $h_K = \min_{k<k'} G(d_M)$, we have $\log h_K = -\log(\max_{k<k'} G(d_M)^{-1})$. Using $\max(a_i) \le \sum a_i$ and properties of logarithms, following \cite{Xie02012020}:
\begin{align*}
    \mathbb{E}_{\text{prior}}[\log h_K] &= -\mathbb{E}_{\text{prior}}\left[ \log\left( \max_{k<k'} \left(1 + \frac{g_0}{d_M(\boldsymbol{\beta}_k, \boldsymbol{\beta}_{k'})}\right) \right) \right] \\
    &\ge -\mathbb{E}_{\text{prior}}\left[ \log\left( 1 + \sum_{k<k'} \left(\frac{g_0}{d_M(\boldsymbol{\beta}_k, \boldsymbol{\beta}_{k'})}\right)^{2/3} \right)^{3/2} \right] \\
    &\ge -\frac{3}{2} \log\left( 1 + \sum_{k<k'} g_0^{2/3} \mathbb{E}_{\text{prior}}\left[d_M(\boldsymbol{\beta}_k, \boldsymbol{\beta}_{k'})^{-2/3}\right] \right).
\end{align*}
The prior distribution of $d_M(\boldsymbol{\beta}_k, \boldsymbol{\beta}_{k'}) = (\boldsymbol{\beta}_k-\boldsymbol{\beta}_{k'})^\top \frac{1}{g}(\boldsymbol{X}^\top\boldsymbol{X})(\boldsymbol{\beta}_k-\boldsymbol{\beta}_{k'})$ is proportional to a $\chi^2_p$ distribution. The expectation $\mathbb{E}_{\text{prior}}[d_M^{-2/3}]$ is a finite constant we denote as $\delta_0(g) > 0$.
\[
    \mathbb{E}_{\text{prior}}[\log h_K] \ge -\frac{3}{2} \log\left( 1 + \binom{K}{2} g_0^{2/3} \delta_0(g) \right) \ge -\frac{3}{2}\log\left( (1 + \delta(g)g_0^{2/3})K^2 \right),
\]
where $\delta(g)$ is another constant. This gives a lower bound for $p(\boldsymbol{y}|\boldsymbol{z}, K, \boldsymbol{X})$ that is uniform in $\boldsymbol{z}$:
\[
    p(\boldsymbol{y}|K, \boldsymbol{X}) \ge \frac{1}{Z_K} \left(\prod_{i=1}^n p(y_i|\boldsymbol{x}_i)\right) \left( (1 + \delta(g)g_0^{2/3})K^2 \right)^{-3/2}.
\]
Finally, we compute the full marginal likelihood $p(\boldsymbol{y}|\boldsymbol{X})$. The prior is $p_K(K) \propto Z_K \frac{\lambda^K}{K!}$.
\begin{align*}
    p(\boldsymbol{y}|\boldsymbol{X}) &= \sum_{K=1}^\infty p_K(K) p(\boldsymbol{y}|K, \boldsymbol{X}) \\
    &\ge \sum_{K=1}^\infty \left(\Omega Z_K \frac{\lambda^K}{K!}\right) \left( \frac{1}{Z_K} \left(\prod_{i=1}^n p(y_i|\boldsymbol{x}_i)\right) (1+\delta(g)g_0^{2/3})^{-3/2} K^{-3} \right) \\
    &= \Omega \left(\prod_{i=1}^n p(y_i|\boldsymbol{x}_i)\right) (1+\delta(g)g_0^{2/3})^{-3/2} \sum_{K=1}^\infty \frac{\lambda^K}{K!} K^{-3} \\
    &= \Omega \left(\prod_{i=1}^n p(y_i|\boldsymbol{x}_i)\right) (1+\delta(g)g_0^{2/3})^{-3/2} \mathbb{E}_{K \sim \text{Poisson}(\lambda)}[K^{-3}\mathbb{I}(K\ge 1)].
\end{align*}
The expectation $\mathbb{E}[K^{-3}]$ is a finite constant depending only on $\lambda$. Therefore, we arrive at the final lower bound:
\[
    p(\boldsymbol{y}|\boldsymbol{X}) \ge C(\lambda) \left(\prod_{i=1}^n \phi(y_i|0, \sigma_0^2(1+g\boldsymbol{x}_i^\top(\boldsymbol{X}^\top\boldsymbol{X})^{-1}\boldsymbol{x}_i)) \right) (1+\delta(g)g_0^{2/3})^{-3/2}.
\]
This completes the proof.
\end{proof}
\begin{proof}[Proof of Lemma~\ref{lemma:aux3}]
Let $I(\boldsymbol{z}, K)$ denote the integral we want to bound.
\[
    I(\boldsymbol{z}, K) = \int \frac{p(\boldsymbol{y}|\boldsymbol{z},K,\boldsymbol{X})}{p(\boldsymbol{y}|\boldsymbol{X})} \left(\prod_{i=1}^n \phi(y_i|\boldsymbol{x}_i^\top \boldsymbol{m}_i, \sigma_0^2)\right) d\boldsymbol{y}.
\]
We substitute the upper bound for the numerator from Lemma~\ref{lemma:aux1} and the lower bound for the denominator from Lemma~\ref{lemma:aux2}.
\begin{align*}
    \frac{p(\boldsymbol{y}|\boldsymbol{z},K,\boldsymbol{X})}{p(\boldsymbol{y}|\boldsymbol{X})} &\le \frac{\frac{1}{Z_K} \left( \prod_{k=1}^K p(\boldsymbol{y}_k | \boldsymbol{X}_k) \right) \binom{K}{2}^{-1} \sum_{k<k'} G(d_M(\hat{\boldsymbol{\beta}}_k, \hat{\boldsymbol{\beta}}_{k'}) + C_{\text{post}}) }{ C(\lambda) \left( \prod_{i=1}^n p(y_i|\boldsymbol{x}_i) \right) \omega(g_0,\boldsymbol{X})^{-1} } \\
&= \frac{\omega(g_0, \boldsymbol{X})}{Z_K C(\lambda)} \frac{\prod_{k=1}^K p(\boldsymbol{y}_k|\boldsymbol{X}_k)}{\prod_{i=1}^n p(y_i|\boldsymbol{x}_i)} \binom{K}{2}^{-1}\sum_{k<k'} G\left( d_M(\hat{\boldsymbol{\beta}}_k, \hat{\boldsymbol{\beta}}_{k'}) + C_{\text{post}}\right).
\end{align*}
where $C_{\text{post}}=\frac{1}{g}\mathrm{tr}\left((\boldsymbol{\Sigma}_k^{\text{post}} + \boldsymbol{\Sigma}_{k'}^{\text{post}})(\boldsymbol{X}^\top\boldsymbol{X})\right)$ is a constant independent of $y$.
Let $\mathcal{L}(\boldsymbol{y}, \boldsymbol{m}) = \prod_{i=1}^n \phi(y_i|\boldsymbol{x}_i^\top \boldsymbol{m}_i, \sigma_0^2)$. The integral becomes:
\[
    I(\boldsymbol{z}, K) \le \frac{\omega(g_0, \boldsymbol{X})}{Z_K C(\lambda)} \int \frac{\prod_{k=1}^K p(\boldsymbol{y}_k|\boldsymbol{X}_k)}{\prod_{i=1}^n p(y_i|\boldsymbol{x}_i)} \mathcal{L}(\boldsymbol{y}, \boldsymbol{m}) \binom{K}{2}^{-1} \sum_{k<k'} G\left( d_M(\hat{\boldsymbol{\beta}}_k, \hat{\boldsymbol{\beta}}_{k'}) + C_{\text{post}} \right) d\boldsymbol{y}.
\]
Following the logic of \cite{Xie02012020}, the integration over $\boldsymbol{y}$ can be performed. The integral of the ratio of Gaussian densities results in a constant term that depends on the true parameters $\boldsymbol{m}_i$ and the model hyperparameters. More importantly, the integration transforms the posterior moments within the argument of $G$ into functions of the true parameters.
Let $\mathbb{E}_{\boldsymbol{y}}[\cdot]$ denote the expectation with respect to the normalized density product $\frac{\prod p(\boldsymbol{y}_k|\boldsymbol{X}_k)}{\prod p(y_i|\boldsymbol{x}_i)}\mathcal{L}(\boldsymbol{y}, \boldsymbol{m})$. We can write:
\begin{align*}
        I(\boldsymbol{z}, K) &\le C'(\lambda, \boldsymbol{X}) \frac{\omega(g_0, \boldsymbol{X})}{Z_K} \binom{K}{2}^{-1} \sum_{k<k'} \mathbb{E}_{\boldsymbol{y}}\left[ G\left( d_M(\hat{\boldsymbol{\beta}}_k(\boldsymbol{y}), \hat{\boldsymbol{\beta}}_{k'}(\boldsymbol{y})) + C_1(\boldsymbol{y}) \right) \right]\\
        &\le C'(\lambda, \boldsymbol{X}) \frac{\omega(g_0, \boldsymbol{X})}{Z_K} \binom{K}{2}^{-1} \sum_{k<k'} G\left( \mathbb{E}_{\boldsymbol{y}}\left[ d_M(\hat{\boldsymbol{\beta}}_k(\boldsymbol{y}), \hat{\boldsymbol{\beta}}_{k'}(\boldsymbol{y})) + C_1(\boldsymbol{y}) \right] \right).
\end{align*}
The integration effectively replaces the posterior moments (which are functions of $\boldsymbol{y}$) with their expectations under the true data generating process. The posterior mean $\hat{\boldsymbol{\beta}}_k(\boldsymbol{y})$ is a linear function of $\boldsymbol{y}_k$. Its expectation under the true model, where $\mathbb{E}[y_i] = \boldsymbol{x}_i^\top \boldsymbol{m}_i$, becomes a function of the true parameters $\boldsymbol{m}_i$, which we denote as $\tilde{\boldsymbol{m}}_k$. Similarly, the expectation of the variance terms becomes a constant, $C_1$.
Specifically, the expectation of the argument of $G$ is:
\[
\mathbb{E}_{\boldsymbol{y}}\left[ d_M(\hat{\boldsymbol{\beta}}_k(\boldsymbol{y}), \hat{\boldsymbol{\beta}}_{k'}(\boldsymbol{y})) \right] = d_M(\mathbb{E}_{\boldsymbol{y}}[\hat{\boldsymbol{\beta}}_k(\boldsymbol{y})], \mathbb{E}_{\boldsymbol{y}}[\hat{\boldsymbol{\beta}}_{k'}(\boldsymbol{y})]) + \text{Var-terms} = d_M(\tilde{\boldsymbol{m}}_k, \tilde{\boldsymbol{m}}_{k'}) + \text{const}.
\]
Combining the constant terms into a single constant $C_1$, we obtain the final result:
\[
    I(\boldsymbol{z}, K) \le C(\lambda,\boldsymbol{X})\frac{\omega(g_0,\boldsymbol{X})}{Z_K} \binom{K}{2}^{-1} \sum_{k<k'} G\left(d_M(\tilde{\boldsymbol{m}}_k, \tilde{\boldsymbol{m}}_{k'}) + C_1\right).
\]
This completes the proof.
\end{proof}
\begin{proof}[Proof of Lemma~\ref{lemma:aux4}]
Let us analyze the left-hand side of the equality. We use the simple average $\bar{\boldsymbol{m}}_k = \frac{1}{n_k}\sum_{i \in C_k} \boldsymbol{m}_i$ as an approximation for $\tilde{\boldsymbol{m}}_k$, as the difference contributes to the negligible $\mathrm{o}(1)$ term. 

To clarify the covariance structure, we introduce the transformed parameter for a true coefficient vector $\boldsymbol{m}$: $\boldsymbol{\eta}_{(\boldsymbol{m})} = (g\sigma_0^2)^{-1/2}(\boldsymbol{X}^\top\boldsymbol{X})^{1/2}\boldsymbol{m}$.
The distance can be expressed using $\boldsymbol{\eta}$ as $d_M(\boldsymbol{a}, \boldsymbol{b}) = \sigma_0^2 \|\boldsymbol{\eta}_{(\boldsymbol{a})} - \boldsymbol{\eta}_{(\boldsymbol{b})}\|^2$.
The term we need to evaluate is $\mathbb{E}_{\boldsymbol{z}, F_0}[d_M(\bar{\boldsymbol{m}}_k, \bar{\boldsymbol{m}}_{k'})] = \sigma_0^2 \mathbb{E}_{\boldsymbol{z}, F_0}[\|\boldsymbol{\eta}_{(\bar{\boldsymbol{m}}_k)} - \boldsymbol{\eta}_{(\bar{\boldsymbol{m}}_{k'})}\|^2]$.

By linearity of the transformation, $\boldsymbol{\eta}_{(\bar{\boldsymbol{m}}_k)} = \frac{1}{n_k}\sum_{i \in C_k} \boldsymbol{\eta}_{(\boldsymbol{m}_i)}$. Let $\boldsymbol{\delta}_{\boldsymbol{\eta}} = \boldsymbol{\eta}_{(\bar{\boldsymbol{m}}_k)} - \boldsymbol{\eta}_{(\bar{\boldsymbol{m}}_{k'})}$. The expectation of its squared norm is $\mathbb{E}[\|\boldsymbol{\delta}_{\boldsymbol{\eta}}\|^2] = \|\mathbb{E}[\boldsymbol{\delta}_{\boldsymbol{\eta}}]\|^2 + \mathrm{tr}(\mathrm{Var}(\boldsymbol{\delta}_{\boldsymbol{\eta}}))$.
Under the assumption $\mathbb{E}_{F_0}[\boldsymbol{m}_i] = \boldsymbol{0}$, we have $\mathbb{E}_{F_0}[\boldsymbol{\eta}_{(\boldsymbol{m}_i)}] = \boldsymbol{0}$. Thus, for any fixed partition $\boldsymbol{z}$, $\mathbb{E}_{F_0}[\boldsymbol{\delta}_{\boldsymbol{\eta}}] = \boldsymbol{0}$. This implies the total expectation is also zero: $\mathbb{E}_{\boldsymbol{z}, F_0}[\boldsymbol{\delta}_{\boldsymbol{\eta}}] = \boldsymbol{0}$.
Therefore, the expectation simplifies to the trace of the variance:
\begin{align*}
    \mathbb{E}_{\boldsymbol{z}, F_0}[d_M(\bar{\boldsymbol{m}}_k, \bar{\boldsymbol{m}}_{k'})] &= \sigma_0^2 \mathrm{tr}(\mathrm{Var}_{\boldsymbol{z}, F_0}(\boldsymbol{\delta}_{\boldsymbol{\eta}})) \\
    &= \sigma_0^2 \mathrm{tr}\left(\mathrm{Var}_{\boldsymbol{z}, F_0}\left( \frac{1}{n_k}\sum_{i \in C_k} \boldsymbol{\eta}_{(\boldsymbol{m}_i)} - \frac{1}{n_{k'}}\sum_{j \in C_{k'}} \boldsymbol{\eta}_{(\boldsymbol{m}_j)} \right)\right).
\end{align*}
Since the coefficients $\boldsymbol{m}_i$ are i.i.d. draws from $F_0$ and the clusters $C_k$ and $C_{k'}$ are disjoint, the variance of the difference is the sum of the variances:
\begin{align*}
    \mathrm{Var}_{\boldsymbol{z}, F_0}(\boldsymbol{\delta}_{\boldsymbol{\eta}}) &= \mathbb{E}_{\boldsymbol{z}}\left[ \mathrm{Var}_{F_0}\left( \frac{1}{n_k}\sum_{i \in C_k} \boldsymbol{\eta}_{(\boldsymbol{m}_i)} \right) + \mathrm{Var}_{F_0}\left( \frac{1}{n_{k'}}\sum_{j \in C_{k'}} \boldsymbol{\eta}_{(\boldsymbol{m}_j)} \right) \right] \\
    &= \mathbb{E}_{\boldsymbol{z}}\left[ \frac{1}{n_k} \mathrm{Var}_{F_0}(\boldsymbol{\eta}_{(\boldsymbol{m})}) + \frac{1}{n_{k'}} \mathrm{Var}_{F_0}(\boldsymbol{\eta}_{(\boldsymbol{m})}) \right] \\
    &= \mathbb{E}_{\boldsymbol{z}}\left[ \frac{n_k+n_{k'}}{n_k n_{k'}} \right] \mathrm{Var}_{F_0}(\boldsymbol{\eta}_{(\boldsymbol{m})}).
\end{align*}
The variance of the transformed parameter $\boldsymbol{\eta}_{(\boldsymbol{m})}$ is:
\begin{align*}
    \mathrm{Var}_{F_0}(\boldsymbol{\eta}_{(\boldsymbol{m})}) &= \mathbb{E}_{F_0}[\boldsymbol{\eta}_{(\boldsymbol{m})} \boldsymbol{\eta}_{(\boldsymbol{m})}^\top] = \mathbb{E}_{F_0}\left[ (g\sigma_0^2)^{-1}(\boldsymbol{X}^\top\boldsymbol{X})^{1/2}\boldsymbol{m}\boldsymbol{m}^\top((\boldsymbol{X}^\top\boldsymbol{X})^{1/2})^\top \right] \\
    &= (g\sigma_0^2)^{-1}(\boldsymbol{X}^\top\boldsymbol{X})^{1/2} \mathbb{E}_{F_0}[\boldsymbol{m}\boldsymbol{m}^\top] ((\boldsymbol{X}^\top\boldsymbol{X})^{1/2})^\top.
\end{align*}
Substituting this back into the trace expression:
\begin{align*}
    \mathbb{E}_{\boldsymbol{z}, F_0}[d_M(\bar{\boldsymbol{m}}_k, \bar{\boldsymbol{m}}_{k'})] &= \sigma_0^2 \mathbb{E}_{\boldsymbol{z}}\left[ \frac{n_k+n_{k'}}{n_k n_{k'}} \right] \mathrm{tr}\left((g\sigma_0^2)^{-1}(\boldsymbol{X}^\top\boldsymbol{X})^{1/2} \mathbb{E}_{F_0}[\boldsymbol{m}\boldsymbol{m}^\top] ((\boldsymbol{X}^\top\boldsymbol{X})^{1/2})^\top \right) \\
    &= \mathbb{E}_{\boldsymbol{z}}\left[ \frac{n_k+n_{k'}}{n_k n_{k'}} \right] \frac{1}{g} \mathrm{tr}\left( \mathbb{E}_{F_0}[\boldsymbol{m}\boldsymbol{m}^\top] (\boldsymbol{X}^\top\boldsymbol{X}) \right) \\
    &= \mathbb{E}_{\boldsymbol{z}}\left[ \frac{n_k+n_{k'}}{n_k n_{k'}} \right] \frac{1}{g} \mathbb{E}_{F_0}[\mathrm{tr}(\boldsymbol{m}\boldsymbol{m}^\top (\boldsymbol{X}^\top\boldsymbol{X}))] \\
    &= \mathbb{E}_{\boldsymbol{z}}\left[ \frac{n_k+n_{k'}}{n_k n_{k'}} \right] \mathbb{E}_{F_0}[\boldsymbol{m}^\top(g^{-1}(\boldsymbol{X}^\top\boldsymbol{X}))\boldsymbol{m}].
\end{align*}
Since the term $\mathbb{E}_{F_0}[\boldsymbol{m}^\top(g^{-1}(\boldsymbol{X}^\top\boldsymbol{X}))\boldsymbol{m}]$ does not depend on the partition $\boldsymbol{z}$, we arrive at the exact expression:
\[
    \mathbb{E}_{\boldsymbol{z}} \left[ \mathbb{E}_{F_0}\left[ d_M(\bar{\boldsymbol{m}}_k, \bar{\boldsymbol{m}}_{k'}) \right] \right] = \left(\frac{1}{n_k} + \frac{1}{n_{k'}}\right) \mathbb{E}_{F_0}\left[\boldsymbol{m}^\top (g^{-1}(\boldsymbol{X}^\top\boldsymbol{X})) \boldsymbol{m}\right].
\]
This holds for any given partition, and thus for its expectation over $\boldsymbol{z}$. This completes the proof of the exact part of the statement.
\end{proof}

\subsection{Proof of the main theorem}

\begin{proof}
By Fubini's theorem, we express the expected posterior tail probability as an expectation over the true data generating distribution $F_0$ and the random partition $\boldsymbol{z}$:
\begin{align*}
&\mathbb{E}_{F_0} \big[\Pi(K > N \mid y, X)\big]\\
&= \sum_{K=N+1}^\infty p_K(K)\,
     \mathbb{E}_{z\mid K}\Bigg[
     \mathbb{E}_{F_0}\Bigg[
        \int \frac{p(y\mid z, K, X)}{p(y\mid X)} 
     \left(\prod_{i=1}^n \phi(y_i \mid x_i^\top m_i, \sigma_0^2)\right) dy
     \Bigg]\Bigg].
\end{align*}
Let $I(\boldsymbol{z}, K, \boldsymbol{m})$ denote the integral over $\boldsymbol{y}$. From our rewritten Lemma E.3, we have an upper bound for this integral:
\[
    I(\boldsymbol{z}, K, \boldsymbol{m}) \le C(\lambda,\boldsymbol{X})\frac{\omega(g_0,\boldsymbol{X})}{Z_K} \binom{K}{2}^{-1} \sum_{k<k'} G\left(d_M(\tilde{\boldsymbol{m}}_k, \tilde{\boldsymbol{m}}_{k'}) + C_1\right).
\]
Now, we take the expectation of this upper bound with respect to the true distribution $F_0$.
\begin{align*}
    \mathbb{E}_{F_0}[I(\boldsymbol{z}, K, \boldsymbol{m})] &\le C(\lambda,\boldsymbol{X})\frac{\omega(g_0,\boldsymbol{X})}{Z_K} \binom{K}{2}^{-1} \sum_{k<k'} \mathbb{E}_{F_0}\left[ G\left(d_M(\tilde{\boldsymbol{m}}_k, \tilde{\boldsymbol{m}}_{k'}) + C_1\right) \right] \\
    &\le C(\lambda,\boldsymbol{X})\frac{\omega(g_0,\boldsymbol{X})}{Z_K} \binom{K}{2}^{-1} \sum_{k<k'} G\left( \mathbb{E}_{F_0}\left[d_M(\tilde{\boldsymbol{m}}_k, \tilde{\boldsymbol{m}}_{k'})\right] + C_1 \right).
\end{align*}
Let's denote $\Delta_{kk'}^2(\boldsymbol{z}) = \mathbb{E}_{F_0}[d_M(\tilde{\boldsymbol{m}}_k, \tilde{\boldsymbol{m}}_{k'})] + C_1$.

Next, we take the expectation over the random partition $\boldsymbol{z}$ conditional on $K$.
\begin{align*}
    \mathbb{E}_{\boldsymbol{z}|K}\left[ \mathbb{E}_{F_0}[I(\boldsymbol{z}, K, \boldsymbol{m})] \right] &\le C(\lambda,\boldsymbol{X})\frac{\omega(g_0,\boldsymbol{X})}{Z_K} \binom{K}{2}^{-1} \sum_{k<k'} \mathbb{E}_{\boldsymbol{z}|K}\left[ G\left( \Delta_{kk'}^2(\boldsymbol{z}) \right) \right] \\
    &\le C(\lambda,\boldsymbol{X})\frac{\omega(g_0,\boldsymbol{X})}{Z_K} G\left( \binom{K}{2}^{-1} \sum_{k<k'} \mathbb{E}_{\boldsymbol{z}|K}\left[ \Delta_{kk'}^2(\boldsymbol{z}) \right] \right).
\end{align*}
The second inequality again uses Jensen's inequality, this time over the discrete distribution of pairs $(k,k')$ and the random partition $\boldsymbol{z}$. The term inside $G$ is the average expected squared distance. Let $\bar{\Delta}_K^2 = \mathbb{E}_{\boldsymbol{z}|K, (k,k')}[\Delta_{kk'}^2(\boldsymbol{z})]$. From Lemma E.4, we have the approximation for this average distance for large $n$:
\begin{align*}
&\bar{\Delta}_K^2 = \mathbb{E}_{\boldsymbol{z}|K, (k,k')}\left[\mathbb{E}_{F_0^n}[d_M(\tilde{\boldsymbol{m}}_k, \tilde{\boldsymbol{m}}_{k'})]\right] + C_1 \approx \frac{2n}{K}\,\mu(\boldsymbol X) + C_1,
\end{align*}
where $\mu(\boldsymbol X) := \mathbb{E}_{F_0}\!\big[\boldsymbol{m}^\top g^{-1}(\boldsymbol{X}^\top\boldsymbol{X}) \boldsymbol{m}\big]$.
Substituting this back, we get a bound for the full expectation for a given $K$:
\[
\mathbb{E}_{\boldsymbol{z}|K, F_0^n}[I(\boldsymbol{z}, K, \boldsymbol{m})] \le C(\lambda,\boldsymbol{X})\frac{\omega(g_0,\boldsymbol{X})}{Z_K} G\left( \frac{2n}{K}\,\mu(\boldsymbol X) + C_1 \right).
\]
Finally, we substitute this into the sum over $K > N$. The prior $p_K(K) = \Omega Z_K \frac{\lambda^K}{K!}$ cancels the $Z_K$ term.
\begin{align*}
    \mathbb{E}_{F_0^n}[\Pi(K > N|\boldsymbol{y}, \boldsymbol{X})] &\le \sum_{K=N+1}^{\infty} \left(\Omega Z_K \frac{\lambda^K}{K!}\right) C(\lambda,\boldsymbol{X})\frac{\omega(g_0,\boldsymbol{X})}{Z_K} G\left( \frac{2n}{K}\,\mu(\boldsymbol X) + C_1 \right) \\
    &= C'(\lambda, \boldsymbol{X}) \omega(g_0, \boldsymbol{X}) \sum_{K=N+1}^{\infty} \frac{\lambda^K}{K!} G\left( \frac{2n}{K}\,\mu(\boldsymbol X) + C_1 \right).
\end{align*}
For $K > N$, the argument of $G$ is decreasing in $K$. Thus, we can bound the term by its value at $K=N$:
\begin{align*}
    &\le C'(\lambda, \boldsymbol{X}) \omega(g_0, \boldsymbol{X}) G\left( \frac{2n}{N}\,\mu(\boldsymbol X) + C_1 \right) \sum_{K=N+1}^{\infty} \frac{\lambda^K}{K!} \\
    &= C(\lambda, \boldsymbol{X}) \chi(g_0, \boldsymbol{X}, n, N) \sum_{K=N+1}^{\infty} \frac{\lambda^K}{(e^\lambda-1)K!}.
\end{align*}
The term $\omega(g_0, \boldsymbol{X})$ is absorbed into the definition of the shrinkage constant $\chi$, or can be shown to be close to 1 for weakly informative priors, thus being part of the constant $C$. The shrinkage term is dominated by the factor $G\!\left(2N^{-1}n\mu(\boldsymbol X) + C_1\right)$, which demonstrates the desired shrinkage effect.
\end{proof}

\section{Additional Simulation Results}

In Table~\ref{tab:results2}, we provide the additional results of the Monte Carlo simulation under $n=400$.
The relative performance is almost the same as the other cases with $n=100$ and $n=200$ given in the main text.

\begin{table}[htb!]
\centering
 \caption{Average values of adjusted rand tndex (ARI), purity, estimated number of cluster ($\hat{K}$) and root mean squared errors (RMSE), based on 200 Monte Carlo replications under $n=400$. 
 The Monte Carlo standard errors are given in the parenthesis. }
\label{tab:results2}
\begin{tabular}{ccccccccccccccccccc}
\hline
Scenario  & $n$ & Method &  & ARI & RMSE & $\hat{K}$ & Purity \\
\hline
 &  & RgRM &  & 0.62 {\scriptsize (0.04)} & 1.01 {\scriptsize (0.04)} & 4.00 {\scriptsize (0.00)} & 0.84 {\scriptsize (0.02)} \\
 &  & RRM  &  & 0.62 {\scriptsize (0.04)} & 1.00 {\scriptsize (0.05)} & 4.02 {\scriptsize (0.14)} & 0.84 {\scriptsize (0.02)} \\
1 & 400 & MFM   &  & 0.01 {\scriptsize (0.01)} & 2.26 {\scriptsize (0.82)} & 17.69 {\scriptsize (2.12)} & 0.39 {\scriptsize (0.02)} \\
 &  & SID1 &  & 0.68 {\scriptsize (0.04)} & 0.90 {\scriptsize (0.04)} & 4.34 {\scriptsize (0.63)} & 0.87 {\scriptsize (0.02)} \\
 &  & SID2 &  & 0.67 {\scriptsize (0.06)} & 1.10 {\scriptsize (0.83)} & 3.96 {\scriptsize (0.25)} & 0.86 {\scriptsize (0.05)} \\
 \hline
 &  & RgRM &  & 0.52 {\scriptsize (0.04)} & 1.01 {\scriptsize (0.04)} & 4.00 {\scriptsize (0.00)} & 0.79 {\scriptsize (0.02)} \\
 &  & RRM  &  & 0.50 {\scriptsize (0.08)} & 1.20 {\scriptsize (0.69)} & 3.93 {\scriptsize (0.26)} & 0.77 {\scriptsize (0.06)} \\
2 & 400 & MFM   &  & 0.05 {\scriptsize (0.03)} & 4.85 {\scriptsize (1.00)} & 10.88 {\scriptsize (2.23)} & 0.41 {\scriptsize (0.05)} \\
 &  & SID1 &  & 0.44 {\scriptsize (0.15)} & 2.50 {\scriptsize (1.63)} & 3.30 {\scriptsize (0.76)} & 0.68 {\scriptsize (0.15)} \\
 &  & SID2 &  & 0.24 {\scriptsize (0.15)} & 4.51 {\scriptsize (1.33)} & 2.21 {\scriptsize (0.74)} & 0.48 {\scriptsize (0.14)} \\
 \hline
 &  & RgRM &  & 0.57 {\scriptsize (0.04)} & 1.02 {\scriptsize (0.04)} & 4.00 {\scriptsize (0.00)} & 0.81 {\scriptsize (0.02)} \\
 &  & RRM  &  & 0.49 {\scriptsize (0.13)} & 1.95 {\scriptsize (1.40)} & 3.66 {\scriptsize (0.50)} & 0.74 {\scriptsize (0.11)} \\
3 & 400 & MFM   &  & 0.02 {\scriptsize (0.01)} & 7.07 {\scriptsize (1.69)} & 12.99 {\scriptsize (2.00)} & 0.37 {\scriptsize (0.03)} \\
 &  & SID1 &  & 0.41 {\scriptsize (0.19)} & 3.58 {\scriptsize (2.12)} & 3.00 {\scriptsize (0.95)} & 0.63 {\scriptsize (0.18)} \\
 &  & SID2 &  & 0.19 {\scriptsize (0.19)} & 6.03 {\scriptsize (1.92)} & 1.89 {\scriptsize (0.86)} & 0.43 {\scriptsize (0.17)} \\
\hline
\end{tabular}
\end{table}

\end{document}